\documentclass{article}
\usepackage[margin=1.1in]{geometry}
\usepackage[round]{natbib}
\usepackage{amsmath, amssymb, enumitem, bm, bbm, multirow, hyperref, graphicx, authblk, tikz, amsthm}
\usepackage[noabbrev]{cleveref}
\usepackage{thmtools,thm-restate}

\makeatletter
\newcommand{\refcheckize}[1]{%
  \expandafter\let\csname @@\string#1\endcsname#1%
  \expandafter\DeclareRobustCommand\csname relax\string#1\endcsname[1]{%
    \csname @@\string#1\endcsname{##1}\@for\@temp:=##1\do{\wrtusdrf{\@temp}\wrtusdrf{{\@temp}}}}%
  \expandafter\let\expandafter#1\csname relax\string#1\endcsname
}
\newcommand{\refcheckizetwo}[1]{%
  \expandafter\let\csname @@\string#1\endcsname#1%
  \expandafter\DeclareRobustCommand\csname relax\string#1\endcsname[2]{%
    \csname @@\string#1\endcsname{##1}{##2}\wrtusdrf{##1}\wrtusdrf{{##1}}\wrtusdrf{##2}\wrtusdrf{{##2}}}%
  \expandafter\let\expandafter#1\csname relax\string#1\endcsname
}
\makeatother

\hypersetup{colorlinks, breaklinks,urlcolor= magenta, linkcolor= blue, citecolor = magenta}

\newcommand{\E}{\mathbb{E}}
\newcommand{\R}{\mathbb{R}}
\newcommand{\N}{\mathbb{N}}

\renewcommand{\P}{\mathbb{P}}
\newcommand{\T}{\mathcal{T}}
\newcommand{\V}{\text{Var}}
\newcommand{\dml}{\text{DML}}

\newcommand{\indep}{\perp \!\!\! \perp}
\allowdisplaybreaks
\makeatletter
\renewcommand\AB@affilsepx{, \protect\Affilfont}
\makeatother

\title{Anytime-Valid Inference for Double/Debiased Machine Learning of Causal Parameters}
\author[1]{Abhinandan Dalal}
\author[2]{Patrick Blöbaum}
\author[2]{Shiva Kasiviswanathan}
\author[2,3]{Aaditya Ramdas}
\affil[1]{University of Pennsylvania}
\affil[2]{Amazon Web Services}
\affil[3]{Carnegie Mellon University}
\date{}

\newtheorem{corollary}{Corollary}[section]
\newtheorem{assumption}{Assumption}[section]
\newtheorem{lemma}{Lemma}
\newtheorem{definition}{Definition}[section]
\newtheorem{example}{Example}

\begin{document}
\maketitle
\begin{abstract}
Double (debiased) machine learning (DML) has seen widespread use in recent years for learning causal/structural parameters, in part due to its flexibility and adaptability to high-dimensional nuisance functions as well as its ability to avoid bias from regularization or overfitting. However, the classic double-debiased framework is only valid asymptotically for a predetermined sample size, thus lacking the flexibility of collecting more data if sharper inference is needed, or of stopping data collection early if useful inferences can be made earlier than expected. This can be of particular concern in large experimental studies with huge financial costs or human lives at stake, as well as in observational studies where the length of confidence intervals do not shrink to zero even with increasing sample size due to partial identifiability of a structural parameter. In this paper, we present time-uniform counterparts to the asymptotic DML results, enabling valid inference and confidence intervals for structural parameters to be constructed at any arbitrary (possibly data-dependent) stopping time. We provide conditions which are only slightly stronger than the standard DML conditions, but offer the stronger guarantee of anytime-valid inference. This facilitates the transformation of any existing DML method to provide anytime-valid guarantees with minimal modifications, making it extremely adaptable and easy to use. We illustrate our procedure using two instances: a) local average treatment effect in online experiments with non-compliance, and b) partial identification of average treatment effect in observational studies with potential unmeasured confounding.
\end{abstract}

\section{Introduction} \label{introduction}

Machine learning (ML) has had a dramatic impact in learning and inferring causal effects in a relatively short period of time. Although machine learning methods are primarily focused on accuracy in prediction, which is a fundamentally different goal than learning of causal effects (\cite{athey2018impact}), a modification in the perspective of applying these ``off-the shelf'' methods (\cite{varian2014big}) allows a smooth marriage of causal inference and ML (\cite{fuhr2024estimating}). Developed by \cite{chernozhukov2018double}, a very important branch of machine learning methods in causal inference is double (debiased) machine learning (DML). Having gained huge popularity in research due to its easy applicability and widespread scope, DML broadly deals itself with the learning of low-dimensional structural parameters in presence of potentially high-dimensional nuisance functions. The flexible framework of DML allows it to leverage a wide array of advanced machine learning methods (\cite{mullainathan2017machine}) to learn complex relationships based on covariates, while still allowing for $\sqrt{n}$ rate of convergence.  

The classical approach on DML, that borrows ideas from semi-parametric inference to create a two step procedure for inferring causal effects, has been traditionally applied in a context of fixed sample size. Although inference is assumed conditional on a large sample size to allow for asymptotics, such a sample size is required to be fixed a priori and cannot be decided after looking at (some amount of) the data. This is an important drawback, especially when conducting experiments or even collection of observational data is considered expensive. Consider the following motivating examples:

\begin{example}\label{coupon}
    A supermarket wishes to roll out a bundle discount on items and is interested in determining the causal effect of the promotion on sales revenue through coupons. Although it assigns treatment and control randomly, not everyone in the treatment group ends up using the coupon. The experiment is expensive, and the supermarket wishes to make a swift decision on whether to terminate or upscale the policy based on the produced effects. 
\end{example}

\begin{example}
    A drug trial assigns a new treatment and placebo randomly to patients, but not all patients end up taking what they are offered due to personal safety concerns or health-neglecting attitude. The drug trial wishes to infer the effect of the drug but wishes to stop as soon as possible because of ethical concerns of human experimentation. \label{drug}
\end{example}

\begin{example}
    A policy is in place for anyone to opt in, but it is voluntary to participate. The policymaker does not have control on who opts for the policy, but is interested to know whether making such a policy compulsory for everyone would be beneficial to the population. The data is not available on the public domain, and is expensive to access (financial costs, privacy concerns, etc.).\label{compulsory}
\end{example}

Example~\ref{coupon} and~\ref{drug} relate to non-compliance (\cite{sagarin2014treatment,rossi2014even}) in experiments. Such non-compliance is fairly common in experiments and leads its way to estimation of local average treatment effects (\cite{cheng2009efficient, mogstad2018identification}) using the treatment assignment itself as an instrumental variable (discussed further in Section~\ref{iv-late}). Informally, an instrumental variable (IV) is a variable that is associated with treatment but affects the outcome only through treatment. Since in Examples~\ref{coupon} and~\ref{drug}, the treatment assignment is predictive of the treatment status, but is unlikely to have an effect on the outcome conditioned on the treatment status, the treatment assignment acts as a valid IV (\cite{hernan2006instruments}).

Example~\ref{compulsory} on the other hand, relates to inference about the average treatment effect (ATE) (\cite{lipkovich2020causal}) in presence of potential unmeasured confounding, and thus is not identified without further assumptions. In such a situation, the ATE cannot be consistently estimated, and what one can realistically hope for is to create bounds on the ATE, a strategy known as partial identification (\cite{tamer2010partial}). Under such a scenario (discussed more in Section~\ref{partial}), the confidence interval length for the parameter of interest would not shrink to zero even with increasing sample size, and it would be all the more imperative to stop collecting data when the researcher deems the inference to be ``good enough''.

All these situations have a commonality that  data collection is expensive, both in experimental and observational situations, and hence a data-adaptive way of deciding to stop data collection would be useful. The traditional central limit theorem based confidence intervals fail under this situation, as it can provide (approximate) validity only at a prefixed (large) sample size $n$. This invalidates any inference made by ``peeking into the data", potentially by a data-dependent stopping time (\cite{johari2017peeking,howard2020time}). Note that, a traditional non-asymptotic $(1-\alpha)$ confidence interval (CI) $C_n$ for a parameter $\theta_0$ at sample size $n$ provides the guarantee that
\begin{equation} \forall n\ge 1, \P(\theta_0\in C_n) \ge 1-\alpha. \label{trad-CI}
\end{equation}
In contrast, if we cannot anticipate in advance when or why we would stop collecting data, we would want to create an interval $\tilde C_\tau$ such that $\P(\theta_0 \in \tilde C_\tau) \ge 1-\alpha$ even when $\tau$ is a stopping time. While $\tilde C_\tau$ can be created to have such coverage for a particular choice of $\tau$, it requires $\tau$ to be pre-regimented, and hence does not allow the flexibility of incorporating data-specific or explorative insights in the stopping decision process. To remedy this, we can bolster our requirement for $\tilde C_\tau$ to be such that 
\begin{equation} \text{for any stopping time }\tau, \ \P(\theta_0\in \tilde C_\tau)\ge 1-\alpha. \label{stop-CS}\end{equation}
It turns out that \eqref{stop-CS} has an equivalent formulation~\citep{howard2021time}:
\begin{equation}
    \P(\forall n\ge 1, \; \theta_0\in \tilde C_n)\ge 1-\alpha, \label{equiv}
\end{equation}
which highlights the distinction with \eqref{trad-CI} in a different manner. An anytime-valid CI, as in \eqref{stop-CS} or \eqref{equiv}, is henceforth referred to as a confidence sequence, in line with the literature (\cite{darling1967confidence, lai1984incorporating, howard2021time}),

Of course, most real-world applications, including almost all applications of DML, involve creating a CI based on an asymptotic version of \eqref{trad-CI}, but an asymptotic anytime-valid version of \eqref{stop-CS} is subtle to define. Luckily, confidence sequences  have recently been appropriately extended to an asymptotic notion recently by \cite{waudby2021time}. We will further review this in Section~\ref{anyvalid-review}. 

\subsection{Related Literature} \label{lit-review}

Sequential experiments and analysis have been a classical problem in statistical literature, to facilitate early stopping. Early examples include the celebrated Wald's sequential probability ratio test (\cite{wald2004sequential}, Chapter~3) and the generalized likelihood ratio approach of \cite{robbins1974expected}. Robbins in particular shifted the focus from Wald's designs of a single stopping rule to instead allowing for continuous monitoring and stopping at any time, possibly not specified in advance. This is not merely of theoretical interest --- if not accounted for correctly, continuous monitoring can lead to substantial inflation of Type-I error (\cite{armitage1969repeated, berman2018p}). Such methods have been seriously appreciated not merely due to their mathematical elegance, but their practical utility. Recent research progress has often been made in industry, examples including risk mitigation in online experimentation (\cite{ham2022design}), rapid software deployment (\cite{lindon2022rapid}) strengthening \textit{canary testing} (\cite{schermann2018we}), multinomial count data modeling for increasing conversion rate (\cite{lindon2022anytime}) from Netflix, flexible monitoring and sample size calculations at Adobe (\cite{maharaj2023anytime}), or usage of sequential testing by Evidently at Amazon Web Services (\cite{evidently}). \cite{turner2023exact} consider the problem of anytime-valid inference in contingency tables, with potential application in labor induction studies (\cite{turner2024generic}). One of the key domains to apply anytime-valid inference is in continuous monitoring of A/B tests \citep{johari2022always}.

A lot of modern advances on this topic have been based on non-asymptotic confidence sequences, such as the constructions of \cite{howard2021time, howard2020time}, with other closely related works on testing (always-valid $p$-values) being \cite{johari2015always} and \cite{balsubramani2015sequential}. The methods of \cite{howard2021time} work under a wide variety of tail conditions on the underlying random variables, and in the sub-Gaussian case improve multi-arm bandit constructions of \cite{garivier2013informational, jamieson2014lil, kaufmann2016complexity}. 

The concept of \emph{asymptotic} confidence sequences (AsympCS), on the other hand, is relatively new, defined and developed by \cite{waudby2021time} as a time-uniform counterpart to the Central Limit Theorem; but it has already seen uptake in several applications, like multi-arm bandits \citep{liang2023experimental} and A/B testing platforms \citep{maharaj2023anytime}. While application of AsympCS on estimating the ATE is well-understood, further efforts are required to extend its use to other causal estimands that benefit from double machine learning estimators. Our paper addresses this issue by introducing a framework similar to the DML assumptions, ensuring an AsympCS construction. This enables asymptotically anytime-valid inference to be applied to more complex and challenging problems.


Double machine learning has received widespread attention in recent years in a wide variety of domains --- economics (\cite{athey2017state}), healthcare (\cite{farbmacher2022causal}), earth sciences (\cite{cohrs2024double}), etc. Developed formally by \cite{chernozhukov2018double}, it outlays the principle of orthogonalization for mitigating regularization bias introduced by modern machine learning methods, a detailed illustration of which is in Section~\ref{dmlreview}. They combine the insights of creating a Neyman-orthogonal function (\cite{neyman1959optimal}) with the use of sample splitting  to relax entropy conditions required for removing plug-in biases.

DML builds on, and is benefited by, classical work on semi-parametric estimation, by channeling ways of removing plug-in bias of complicated nonparametric functionals. The classical literature focuses on extracting $\sqrt{n}$-consistent estimators admitting a central limit theorem, when nuisance parameters are estimated by conventional methods (e.g.,~\cite{levit1976efficiency, bickel1993efficient, robinson1988root, neweysemi1990, newey1994asymptotic, andrews1994asymptotics, newey1998undersmoothing, newey2004twicing, chernozhukov2020locally}, etc.). Neyman's orthogonality condition (\cite{neyman1959optimal}), defined in \eqref{orthogonal-pure}, is a pivotal condition beyond semiparametric learning theory, for instance in optimal testing and adaptive estimation and targeted learning. In conjunction with Neyman orthogonality, \cite{andrews1994asymptotics}  uses Donsker conditions (\cite{vandervaart}) to establish an equicontinuity condition 
for an estimating function $\psi$, with $\theta_0$ and $\eta_0$ being the true values of the parameter of interest and nuisance parameter(s) respectively. Such Donsker conditions are however a restrictive tool unsuitable for high-dimensional nuisance parameters. The strong entropy conditions required for controlling biases introduced by estimating and evaluating on the same data can be bypassed by cross-fitting via sample-splitting, which not only make proofs simpler, but also acts as a pragmatic approach to counter the overfitting/high-complexity challenges of highly adaptive ML methods (\cite{belloni2010lasso, belloni2012sparse, robins2008higher, robins2013new, robins2017minimax}). 

The construction of Neyman-orthogonal functions is  related to the derivation of influence functions (\cite{tsiatis2006semiparametric}), which ensure second order bias on the estimation of nuisance parameters (\cite{semenova2018orthogonal}), and semi-parametric efficiency under correctly specified models (\cite{newey1994asymptotic}). While influence functions may not exist when pathwise differentiability of the functional of interest is violated (usually via lack of finite variance of the Riesz representor (\cite{hines2022demystifying}), insights derived from them may still be used to ensure second-order bias via pseudo-outcome construction (\cite{kennedy2017non, yang2023forster}), and hence constitute a crucial piece of DML literature. Often of particular interest are influence functions of parameters that admit a second-order bias of the form of a mixed bias, that allow double-robust estimation under possible misspecification (\cite{robins2008higher}). \cite{ghassami2022minimax} also studied a general class of functions in this regard that ensure mixed bias in the form of a product. 

Such a rich literature on influence functions or its DML counterpart lay a fertile ground for extending these results to it's strongly consistent counterparts. Building on the groundwork of anytime-valid inference, this paper positions itself at the intersection of these two avenues, and hopes to develop the framework for strongly consistent asymptotic inference for the DML literature.


\subsection{Our Contributions}

We delineate our contributions in this paper as the following:
\begin{list}{{\bf (\roman{enumi})}}{\usecounter{enumi}
\setlength{\leftmargin}{4ex}
\setlength{\listparindent}{4ex}
\setlength{\parsep}{0pt}}
 \item \textit{Motivational}: While the usual sequential A/B testing framework only considers ATE estimation, the generalized DML framework, which has received widespread uptake in recent years in academic and industrial research alike, has been relatively untouched by anytime-valid inference. While DML can be used in trickier identification situations (like LATE, ATT, partial identification etc.), we motivate its usage in continuously monitored experiments as well as observational studies in Section~\ref{introduction} and in the empirical analysis of Section~\ref{empirical}, and elucidate a formal framework of achieving DML guarantees along with anytime-validity. 
 
 \item \textit{Technical}: While the general recipe for confidence sequence construction under a strongly valid linear asymptotic representation has been provided in \cite{waudby2021time}, it is unclear how such a representation would hold in complicated situations, particularly in the DML context. We address this gap by pin-pointing the \cite{chernozhukov2018double}-style assumptions on the underlying problem and estimation problem, thus allowing users a formal framework to verify and deploy suitable techniques for anytime-valid inference under more sophisticated DML algorithms.
 
 \item \textit{Simplifying}: Last but not the least, we illustrate the simplicity of constructing confidence sequences in complicated causal settings. In fact, while obtaining orthogonal functions and applying the DML framework can sometimes be daunting, the incremental effort to obtain anytime-valid inference is very minimal. Thus, our inferential methods are hardly any more restrictive than the original DML techniques; but providing an additional confidence guarantee and the flexibility that comes with it. We also formalize a framework to obtain Neyman-orthogonal functions for estimation of parameters in Theorem~\ref{pseudooutcome} which, to our knowledge, has not been formalized before in the literature in such generality; thus allowing second order bias in estimating parameters even if they don't admit a valid influence function. 
 \end{list}

The paper is organized as follows --- in Section~\ref{review} we review relevant concepts of anytime-valid inference (Section~\ref{anyvalid-review}) and double machine learning (Section~\ref{dmlreview}), with old and new construction techniques for Neyman-orthogonal functions in Section~\ref{construction}. Section~\ref{main} outlines the assumptions on the estimators and the underlying data regularity to facilitate construction of an AsympCS, and shows how one can construct such sequences. Section~\ref{partial} and Section~\ref{iv-late} considers two non-trivial applications beyond vanilla A/B testing, viz. partial identification in observational studies, and non-compliance in online experiments respectively, and illustrates how our framework is applicable in such situations. 
We evaluate the performance of our method through simulation experiments in Section \ref{sims}, and real data applications in (i) program evaluation via the STAR program in Section~\ref{star} and (ii) differentially expressed genes in microarray analysis under potential unmeasured confounding in Section \ref{gene}, concluding with a discussion in Section~\ref{discuss}. 

All $O(\cdot)$ and $o(\cdot)$ notation in this paper shall be interpreted in the almost sure sense unless otherwise specified.

\section{Review of Relevant Concepts} \label{review}
\subsection{Asymptotic Anytime-valid Inference} \label{anyvalid-review}

We continue this discussion from Section~\ref{introduction}, and adapt most of the concepts from the foundations laid by \cite{waudby2021time}. While anytime-valid inference, by virtue of it being valid \textit{anytime}, is naturally non-asymptotic, it is limited in it's applicability. Most classical procedures for anytime-valid inference require knowledge of the parametric distribution family of the data generating procedure (eg: \cite{darling1968some, kaufmann2021mixture}), or create wider intervals to ensure nonparametric non-asymptotic coverage (\cite{howard2021time}). To contrast this with the fixed sample regime, consider creating confidence intervals for the mean for a well-behaved distribution family --- it either requires knowledge of the distribution family to compute quantiles, or creates too wide bounds based on nonparametric Chebyshev or Chernoff bounds. On the other hand, an asymptotic analysis like the Central Limit Theorem allows one to have best of both worlds --- it creates reasonably narrow intervals with asymptotically valid coverage, which is often good enough for real applications. \cite{waudby2021time} extended this notion to the concept of anytime-valid inference, as summarized below. 

A $(1-\alpha)$-confidence sequence for a real-valued parameter of interest $\theta_0$ is given by $\{\tilde C_n^\star\}_{n\ge 1}$, where $\tilde C_n^\star = [L_n^\star, U_n^\star]$ such that \eqref{equiv} holds for $\tilde C_n^\star$. The definition of an asymptotic confidence sequence $\tilde C_n$, or AsympCS for short is predicated on the existence of a such a nonasymptotic confidence sequence $\tilde C_n^\star$:

\begin{definition}[Asymptotic Confidence Sequences] \label{asymp-cs}
     $\{\tilde C_n\}_{n\ge 1} = \{[\tilde L_n, \tilde U_n]\}_{n\ge 1}$ is a $(1-\alpha)$ AsympCS for a parameter of interest $\theta_0$ if there exists a (maybe unknown) non-asymptotic $(1-\alpha)$ confidence sequence $[L_n^\star, U_n^\star]_{n\ge 1}$ for $\theta_0$ such that $\P(\forall n\ge 1, \theta_0\in [L_n^\star, U_n^\star])\ge 1-\alpha$ and ${L_n^\star}/{\tilde L_n}\overset{a.s.}{\to}1 \text{ and }{U_n^\star}/{\tilde U_n}\overset{a.s.}{\to}1,$ ie, the end points of the nonasymptotic confidence sequence can be almost surely precisely approximated by $\tilde L_n$ and $\tilde U_n$. 
    Furthermore we say that $\tilde C_n$ has an approximation rate $a_n$ if $\max\{L_n^\star - \tilde L_n, U_n^\star -\tilde U_n\} = O(a_n)$.
\end{definition}

Note that Definition~\ref{asymp-cs} has intuitive motivation from the fixed sample case. To see this, consider $\{\hat\mu \pm z_\alpha\hat\sigma/\sqrt n\}$, which is an asymptotic confidence interval for a $\mathcal{N}(\mu,\sigma^2)$ family with unknown parameters $\mu$ and $\sigma^2$. This approximates a valid interval $\{\hat\mu\pm z_\alpha \sigma/\sqrt n\}$, albeit with an unknown parameter $\sigma$ which exists unbeknownst to the researcher, and in that lies the root of its asymptotic validity. Similarly the celebrated Gaussian mixture boundary results for partial sums of iid Gaussian random variables $Z_1,Z_2,\cdots \sim \mathcal{N}(\mu,\sigma^2)$ (\cite{robbins1970statistical}) allows one to create an exact $(1-\alpha)$-CS of the form $$\hat\mu \pm \sigma\sqrt{\frac{2(n\rho^2 + 1)}{n^2\rho^2}\log\left(\frac{\sqrt{n\rho^2+1}}{\alpha}\right)}$$ for any $\rho>0$, with knowledge of $\sigma^2$. Definition~\ref{asymp-cs} extends this property to confidence sequences by allowing replacement of $\sigma^2$ with $\hat\sigma^2$, the sample variance which is known to be an almost sure estimator of $\sigma^2$. Although extensions of this nature departing from the Gaussian behavior have been known, it requires explicit knowledge of the tail behavior of estimators being used (like sub-Gaussian, sub-exponential etc.), which is complicated to pinpoint while using sophisticated machine learning methods. In this regard Definition~\ref{asymp-cs} truly addresses the gap that existed between the classic central limit theorem (CLT) and its time uniform counterpart, and hence we admit it as our go-to definition. For multivariate parameters, we consider Asymptotic Confidence Regions (\cite{waudby2021time}):

\begin{definition} [Multivariate Asymptotic Confidence Sequence]
An $\R^d$ valued random set $\{\tilde C_n\}_{n\ge 1}$ is a $(1-\alpha)$ asymptotic confidence sequence (AsympCS) for a parameter $\theta_0\in \R^d$ if there exists a non-asymptotic $(1-\alpha)$ confidence sequence $\{C_n^\star\}_{n\ge 1}$ for $\theta_0$ such that $\P(\forall n\ge 1, \theta_0\in C_n^\star)\ge 1-\alpha$ and $\Lambda(\tilde C_n\Delta C_n^\star)/\Lambda(C_n^\star) \to 0$ almost surely, where $\Delta$ denotes set-difference and $\Lambda$ the Lebesgue measure on $\R^d$.
\end{definition}

An important property an AsympCS is expected to have is uniform Type-I error control:
\begin{definition}[Asymptotic time-uniform coverage] \label{coverage}
    A sequence of AsympCSs $\{\tilde C_n(m)\}_{n\ge 1}$, $m\in \N$ is said to have an asymptotic time uniform $(1-\alpha)$-coverage for $\theta_0$ if 
    \begin{equation}
        \lim_{m\to\infty} \P(\forall n\ge m, \theta_0\in \tilde C_n(m))\ge 1-\alpha. \label{type1}
    \end{equation}
\end{definition}
Several of the constructions in~\cite{waudby2021time} satisfy~\eqref{type1} with equality. Condition \eqref{type1} essentially says that if one starts peeking at the data late enough, i.e.\ after a large sample size $m$, one expects to cover the true parameter throughout with high probability. While \eqref{type1} is an important property to have for asymptotic confidence sequences, it doesn't usually suffice- see \cite{waudby2021time} for a detailed discussion. However, the confidence sequences we shall consider in this paper satisfies both Definition~\ref{asymp-cs} and Definition~\ref{coverage} and hence avoid the conundrum of choosing one versus the other. 


\textbf{Remark} (\textit{Running Intersections}): It might be worthwhile to mention that a feature of non-asymptotic confidence sequences that are not enjoyed by it's batch counterpart is that, if $\{C_n^\star\}_{n\ge 1}$ is a confidence sequence, then so is $\{C_n^\cap\}_{n\ge 1}$, where $C_n^\cap = \cap_{k\le n} C_k^\star$. A similar property is enjoyed by AsympCS's that satisfy \eqref{type1}: If $\{\tilde C_n(m)\}$ satisfy $\eqref{type1}$ for $m=1,2,\cdots$, then so does $\tilde C_n^\cap(m)= \cap_{m\le k\le n} \tilde C_k(m)$. This is because $\lim_{m\to\infty}\P(\exists n\ge m: \theta_0\notin \tilde C_n^\cap(m)) = \lim_{m\to\infty}\P(\exists n\ge m, m\le k\le n: \theta_0\notin \tilde C_k(m)) = \lim_{m\to\infty}\P(\exists k\ge m: \theta_0\notin \tilde C_k(m)) \ge (1-\alpha)$. This is particularly advantageous in case of ML based estimators which are known for its instability (\cite{molybog2023theory}), due to various reasons including numerical precision, sensitivity to initialization and early layers, and size of hyperparameter space (\cite{hammoudeh2024training, lones2021avoid, smith2017cyclical}). Thus, here one can use such running intersections to create confidence sequences monotonic in sample size (of course, only after a sufficiently large initial peeking time $m$) and still maintain asymptotic time uniform coverage as in \eqref{type1}.

\subsection{Double/Debiased Machine Learning} \label{dmlreview}

Having reviewed the concepts of anytime-valid inference, we turn our attention to double-debiased machine learning. The primary success of DML techniques is to overcome the regularization bias that is obtained from machine learning estimators (\cite{cucker2002best}). For instance, even in a simple partially linear model
\begin{align}
    Y &= A\theta_0 + m_0(X) + \varepsilon_1, \ \ &\E[\varepsilon_1|X,A] &= 0, \label{plr1}\\
    A &= e_0(X) + \varepsilon_2, \ \ \ &\E[\varepsilon_2|X]&= 0 \label{plr2},
\end{align}
a typical estimator for~\eqref{plr1} will be of the form $A\hat\theta + \hat m$.\footnote{Unless otherwise specified, we use $\theta,\eta,m$ etc. to represent a general parameter, $\hat\theta,\hat\eta, \hat m$ etc. for their estimates, and $\theta_0,\eta_0, m_0$ etc. for the true parameters.} For the ease of exposition, consider a sample-splitting framework where $\hat m$ is calculated independently from one half of the sample $I^c$, and $\theta_0$ is learnt using a different part $I$ using the learnt $\hat m$ and the OLS estimator. Even in such a situation, \cite{chernozhukov2018double} illustrates that 
$$\sqrt{n}(\hat\theta - \theta_0) = \left(\dfrac{2}{n}\sum_{i\in I} A_i^2\right)^{-1}\left[\dfrac{2}{\sqrt n}\sum_{i\in I}A_i\varepsilon_{1i} + \dfrac 2{\sqrt{n}}\sum_{i\in I}e_0(X_i)(m_0(X_i) - \hat m(X_i)) + o_p(1)\right].$$
The first (inverse) factor converges to a constant, and the first term in the second (square bracket) factor admits a central limit under mild conditions, but the second term in the second factor is not bounded in probability and generally diverges. This is because machine learning estimators, in an attempt to control the variance of $\hat m$, force it away from the true $m_0$ in expectation.
This problem can however be mitigated by extracting out the effect of $X$ from $A$ by \eqref{plr2}. Using the same framework as before, learning $\hat m$ and $\hat e$ from $I^c$ and using the DML estimator $\hat\theta_\dml = \sum_{i\in I} \hat \varepsilon_{2i}(Y_i - \hat m(X_i)) / \sum_{i\in I} \hat \varepsilon_{2i} A_i$ with $\hat\varepsilon_{2i} := A_i - \hat e(X_i)$, one may obtain that 
\begin{align} \sqrt{n}(\hat\theta_\dml - \theta_0) &= \E[\varepsilon_{2}^2]^{-1}\left(\dfrac{2}{\sqrt n}\sum_{i\in I} \varepsilon_{1i}\varepsilon_{2i} + \dfrac{2}{\sqrt n}\sum_{i\in I} (\hat m(X_i) - m_0(X_i))(\hat e(X_i) - e_0(X_i)) + r^\star\right). \label{plr-rep}
\end{align}
Now the first term in parentheses admits a central limit. The second term is less problematic as it is now the product of two estimation errors, and hence can vanish at rate $\sqrt{n} n^{-(\psi_m +\psi_e)}$, even when the individual rates $\sqrt{n}n^{-\psi_m}$ and $\sqrt{n}n^{-\psi_e}$ diverge to $\infty$. Thus, if it can be verified by $r^\star$ is sufficiently well-behaved, ensured by either equicontinuity conditions as discussed in Section~\ref{lit-review}, or by sample-splitting methods as illustrated by \cite{chernozhukov2018double}, one can expect to obtain an asymptotic linear representation --- facilitating construction of confidence intervals via CLT if the approximation holds only weakly, or anytime-valid intervals in case of strong approximation.

The key idea in this regard was the partialling out of $X$ from $A$, an idea formally known as orthogonalization. Introduced by \cite{neyman1959optimal, neyman1979c} to the semiparametric literature, it is a crucial idea for not only DML but a wide array of statistical learning, as alluded to in Section~\ref{lit-review}. To define Neyman orthogonality, suppose that we have access to an iid stream of data $\{W_i\}_{i\ge 1}$ with some law $W$ on the measurable space $\mathcal{W}$. Let $\theta\in \Theta$ be our parameter of interest, and $\eta:\mathcal{W}\to \R^{d_0}$ be nuisance functions, with their collection being represented by a convex set $\mathcal{T}$. The true parameter values are denoted by $\theta_0\in\Theta$ and $\eta_0\in \mathcal{T}$ respectively, and let $\psi: \mathcal{W}\times\Theta\times \mathcal{T}\to \R^d$ be a function such that 
\begin{equation} 
    \E[\psi(W;\theta_0,\eta_0)] = 0. \label{moment}
\end{equation}
Condition~\eqref{moment} ensures that $\theta_0$ is identified from the data distribution, and can generally be obtained from a modification of the identifying equation. 

\begin{definition}[Neyman Orthogonality]
    The function $\psi:\mathcal{W}\times\Theta\times \mathcal{T}\to\R^d$ satisfies the Neyman orthogonality condition at $(\theta_0,\eta_0)$ on the nuisance realization set $\mathcal{T}_n\subseteq \mathcal{T}$ if it satisfies \eqref{moment}, admits a pathwise (Gateaux) derivative  $\frac{\partial}{\partial r}\psi(W;\theta_0,\eta_0 + r(\eta -\eta_0))$ for $r\in [0,1)$ and $\eta\in \mathcal{T}$, and satisfies
    \begin{equation}
        \dfrac{\partial }{\partial r}\E[\psi(W;\theta_0, \eta_0 + r(\eta-\eta_0))] = 0 \label{orthogonal-pure}
    \end{equation}
for all $r\in [0,1)$ and $\eta\in \mathcal{T}_n \subseteq \mathcal{T}$.
\end{definition}
If $\eta$ is a finite dimensional vector then $\mathcal{T}_n = \mathcal{T}$ suffices.
While \cite{chernozhukov2018double} lays out conditions for a linear representation like \eqref{plr-rep} to hold in probability for a Neyman orthogonal function $\psi$ as in \eqref{orthogonal-pure}, it is non-trivial to extend these conditions for an almost sure representation --- which is necessary to apply the recipe of \cite{waudby2021time} to obtain AsympCS for the structural parameters. In this paper, we address the gap by formulating the exact Chernozhukov-style DML conditions for such an almost sure representation to hold, and hence constructing asymptotic confidence sequences. We must also mention that the conditions for a strong approximation are remarkably similar to \cite{chernozhukov2018double}'s weaker representation, and hence with only a slightly incremental cost can one extend the guarantee of DML confidence intervals to a much stronger anytime-valid guarantee, with all of its benefits and flexibility.

\subsection{Construction of Neyman-orthogonal Functions} \label{construction}

Before moving on to the time-uniform counterpart DML algorithms, we should address the major bottleneck of obtaining Neyman-orthogonal functions. Indeed the construction of a Neyman-orhogonal function $\psi$ satisfying \eqref{orthogonal-pure} is hard, and lies at the heart of DML methods. For instance, one often obtains $\theta\in\Theta$ and $\xi\in \Xi$ (the parameter of interest and a nuisance parameter respectively) by an optimization problem $$\max_{\theta\in\Theta,\xi\in\Xi}\E[l(W;\theta,\xi)],$$ which under usual regularity conditions satisfy the first order conditions $$\E[\partial_\theta l(W;\theta,\xi)] = 0\text{ and }\E[\partial_\xi l(W;\theta,\xi)] =0.$$ Then, taking $\eta = (\xi,\mu)$, one may define $\psi(W;\theta,\eta) = \partial_\theta l(W;\theta_0,\xi) - \mu \partial_\xi l(W;\theta,\xi)$, with $J_{\theta\xi} - \mu J_{\xi\xi} = 0$ where $J_{ab} = \partial_{a,b} \E[\partial_{ab} l(W;\theta,\xi)]|_{\theta_0,\xi_0}$ for $a,b\in \{\theta,\xi\}$. It is easy to verify that under suitable regularity and invertibility conditions, $\psi$ constitutes a Neyman-orthogonal function. Of course, this construction is not new, as illustrated in the context of constructing an efficient score when $l$ is the true log-likelihood (\cite{neyman1959optimal}), or in more general cases by \cite{chernozhukov2018double}. The latter also considers a wide array of case-specific instances for construction of a Neyman-orthogonal $\psi$, which we do not reiterate for considerations of space.

While it may be clear from~\eqref{moment} and~\eqref{orthogonal-pure} that Neyman-orthogonality is a property of the function $\psi$ and the true value of the nuisance parameter $\eta_0$ and not determined by a modeling choice for the parameter $\theta$, it might still be a prudent choice to pursue the semiparametric efficiency approach to construct such an orthogonal function $\psi$. Such a calculation leads to an influence function adjustment to the original function, and is often related to achievement of semi-parametric efficiency for estimating $\theta_0$ under the specific semiparametric model (\cite{chernozhukov2020locally, newey1994asymptotic}).

The route to obtaining such a functional usually involves that  the parameter (or functional) of interest $\theta_0$ is pathwise-differentiable with respect to the underlying distribution, so that it may admit an influence function. To illustrate, one may consider a one dimensional mixture model of the underlying distribution $P$ and a perturbing distribution $\tilde P$, then one may use the parametric submodel given by $P_r = P + r(\tilde P - P)$, $r\in [0,1)$, often referred to as the regular parametric submodel (\cite{hines2022demystifying}). To derive an influence function for $\theta$, one requires that $\left.\frac{d\theta(P_r)}{dr}\right|_{r=0}=\lim_{r\downarrow 0} \frac{\theta(P_r) - \theta(P)}{r}$, that is the pathwise derivative of $\theta(P_r)$, to exist and have finite variance. Such a pathwise differentiability criteria is very crucial for semiparametric efficiency, and allows one to write 
$\left.\dfrac{d\theta(P_r)}{dr}\right|_{r=0} = \E_{\tilde P}[\zeta_R(W,P)]$, where $\zeta_R\in L^2$ is a unique `representator' of the pathwise derivative, allowed for by the Riesz representation theorem (\cite{bickel1993efficient, laan2003unified, tsiatis2006semiparametric}). In such a situation, one considers the estimator $\theta(\hat P_n) +\frac 1n\sum_{i=1}^n \zeta_R(W_i,\hat P_n)$, where $\hat P_n$ denotes the empirical distribution, as the estimating equation of the form \eqref{moment}.

However, the parameter of interest $\theta_0$, although identified by a function $\phi$, may not be pathwise differentiable. Examples include the conditional average treatment effect (CATE) under continuous covariates (\cite{rubin2007doubly}), or the potential outcome mean under a continuous treatment model (\cite{rubin2005general, kennedy2017non}). However, orthogonal functions satisfying \eqref{orthogonal-pure} can be obtained by instead considering the expectation of the functional $\theta_0$ with respect to a marginal density on the continuous variable. Generally, these expectations are pathwise differentiable and hence admit an influence function, and thus may be utilized to extract a pseudo-outcome to satisfy Neyman-orthogonality (\cite{kennedy2017non, yang2023forster}). The general result for formalizing the approach for deriving a pseudo-outcome is given as follows, the proof of which can be found in Appendix~\ref{pseudoproof}. 

\begin{restatable}{theorem}{first}
\label{pseudooutcome}
    Consider a parameter of interest $\theta\in\Theta$ that is identified by a function $\phi\in L^2$ satisfying $\E[\phi(W;\theta_0,\beta)|X] = 0$, where $\beta\in\mathcal{B}$ is some nuisance parameter (possibly infinite dimensional), $W$ is the observed data, and $X$ is some observed random variable. Suppose that there exists a function $R(W;\theta,\eta)\in L^2$, with $\eta = (\beta, \gamma(\beta))$, such that for any regular parametric model $P_r$ with parameter $r$ such that $P_0 = P$ and likelihood score $S(\cdot)$, we have
    \begin{equation}
        \left.\dfrac{\partial\E[\phi(W;\theta_0,\beta_0 + r(\beta - \beta_0))]|X]}{\partial r}\right|_{r=0} = \E[R(W;\theta_0,\underbrace{\beta_0, \gamma(\beta_0)}_{\eta_0})S(W)|X], \label{pseudocorrection}
    \end{equation}
    with $\E[R(W;\theta_0,\eta)|X] = 0$. Then, $\psi(W;\theta_0, \eta) = \phi(W;\theta_0,\beta) + R(W;\theta_0,\eta)$ satisfies  Neyman orthogonality~\eqref{orthogonal-pure}.
\end{restatable}

Note that the function $\E[\phi(W;\theta_0, \beta)|X=x]$ may not be pathwise differentiable as it may not admit a second moment (and hence not an influence function), but Theorem~\ref{pseudooutcome} bypasses such a requirement for the construction of Neyman-orthogonal functions. \cite{yang2023forster} considers the special case of this Theorem~\ref{pseudooutcome} in the case of counterfactual regressions, and provide a wide variety of instances of the construction of such a $R$ function and hence a corresponding pseudo-outcome. 


\section{Construction of Asymptotic Confidence Sequences in the DML Framework} \label{main}

Suppose that we get to observe a sequence $\{W_i\}_{i\ge 1}$ iid with the same law as $W$ on a measurable space $\mathcal{W}$. Let nuisance functions be denoted by $\eta: \mathcal{W}\to \R^{d_0}$, and $\mathcal{T}$ be the collection of these functions. Denote $\theta$ as the parameter we actually care to estimate, taking value in $\Theta$. 

The true parameter and true nuisance are denoted by $\theta_0\in \Theta$ and $\eta_0\in \mathcal{T}$ respectively, and with $\hat\theta$ and $\hat\eta$ denoting their estimates, which shall be discussed in detail next. Let $\psi:\mathcal{W}\times\Theta\times\mathcal{T}\to\R^d$ be a function satisfying \eqref{orthogonal-pure}. Also, let $\delta_n$ be a sequence of positive constants converging to 0. We characterize the uncertainty due to the estimation of nuisance parameters by $\delta_n$, and generally, under high-dimensional parameters, we expect $\delta_n$ to be slower than the parametric rate, i.e., $\delta_n\ge n^{-\frac 12}$, and therefore the bias in estimation of $\theta_0$ or asymptotic approximation by a suitable normal distribution is dominated by the uncertainty in estimation of nuisance parameters (\cite{yang2023forster}). Let $T$ be the number of samples used to evaluate, and $T' = n-T$ is the number of samples required to learn the nuisance parameters till time $n$.

To estimate $\hat\eta$, we take a $K$-fold random partition of $[n] = \{1,\dots, n\}$, denoted by $\{I_k\}_{1\le k\le K}$, with the size of each $I_k$ roughly equalling $n/K$. Such a partition may be obtained by sequentially assigning each data point to one of $I_k$ at random, or by looping through each $I_k$ over $k$ one at a time till all the data points upto time $n$ are exhausted, etc. Nevertheless, we expect the size of each $I_k$ to be good enough for a fair estimation of $\hat\eta$, in a manner made precise by Assumption~\ref{regularity}. 

Also, we use the notation $\P_tf(Z) = \frac 1t\sum_{i=1}^t f(Z_i)$ for the empirical average operator, while the conditional average operator is denoted by  $\P \hat f(Z) = \E[\hat f(Z)|\{Z_i'\}_{i\in I_1^c}]$, where $\hat f$ is learnt from $\{Z_i'\}_{i\in I_1^c}$. 

Define the sample splitting estimator $\hat\theta$ as the solution to 
$$\P_T[\psi(W,\theta,\hat\eta)] = 0,$$ where $T = |I_1|$, and $\hat\eta$ is estimated from $I_{-1} := \{W_i\}_{i\in I_1^c}$. For convenience, we shall denote $I_1$ as simply $I$. Proof of all results in this section are collected in Appendix~\ref{mainproofs}.

\begin{assumption}
\label{neyman-orthogonal} (\textit{Linear Scores with Approximate Neyman Orthogonality}) For all $n\ge 3$, assume that the following conditions hold: 
\begin{itemize}
    \item[(a)] The true parameter value $\theta_0$ satisfies moment conditions in \eqref{moment}. 
    \item[(b)]The score function is linear in $\theta$, that is, there exist functions $\psi^a$ and $\psi^b$ such that 
    \begin{equation}
        \psi(W;\theta,\eta) = \psi^a(W;\eta)\theta + \psi^b(W;\eta). \label{linear_split} 
    \end{equation}
    \item[(c)] The map $\eta\mapsto \E[\psi(W;\theta,\eta)]$ is twice continuously Gateuax differentiable on $\mathcal{T}$.
    \item[(d)] The score function $\psi$ satisfies the Neyman near-orthogonality condition at $(\theta_0,\eta_0)$ with respect to the nuisance realization set $\T_n\subset \T$, ie, 
    \begin{equation} \lambda_n :=\sup_{\eta\in \T_n}\left\|\left.\dfrac{d}{dr}\E[\psi(W;\theta_0,\eta_0 + r(\eta-\eta_0))]\right|_{r=0}\right\|\leq 
    \delta_n \sqrt{\dfrac{ \log\log n}{n}}.
    \label{orthogonal}
    \end{equation}
    \item[(e)] The identification holds; ie, the singular values of the matrix $J_0 := \E[\psi^a(W;\eta_0)]$ lie between $c_0$ and $c_1$, where $c_0\le c_1$ are finite positive constants.
\end{itemize}
\end{assumption}

Of course, Assumption~\ref{neyman-orthogonal}(a) is the one ensuring that $\theta_0$ can be identified and estimated from the population moment equation given by $\psi$. Condition~\eqref{linear_split} in Assumption~\ref{neyman-orthogonal}(b) essentially assumes that $\theta$ can be inferred back from $\psi$ as a solution to a set of linear equations. While this is a strong assumption, most causal estimands, like average treatment effect (ATE), local average treatment effect (LATE), conditional average treatment effect (CATE), average treatment effect on treated (ATT) etc.\ can be identified and estimated from a $\psi$ function that admits a linear representation of the form in \eqref{linear_split}. The Assumptions (c) and (d), generalize the Neyman orthogonality condition in \eqref{orthogonal-pure} to a near-orthogonality condition in cases where exact orthogonality is too restrictive. Such a relaxation may be useful for quasi-likelihood or GMM settings, as illustrated in \cite{chernozhukov2018double}. Note that the near-orthogonality condition in \eqref{orthogonal} allows for an additional factor of $\sqrt{\log\log n}$ not allowed for in \cite{chernozhukov2018double}, but is not really a relaxation as the nuisance realization set $\mathcal{T}_n$ needs to contain $\hat\eta$ almost always, as shall be clarified in Assumption~\ref{regularity}. Assumption (e) is a simple assumption that in conjunction with \eqref{linear_split} allows $\theta_0$ to be indeed identified from $\psi$. 

In addition, to Assumption~\ref{neyman-orthogonal}, we need a second set of regularity assumptions on the behavior of $\psi$ in $\mathcal{T}_n$, to ensure that all plug-in biases are controlled and we may obtain an asymptotically linear representation to hold almost surely. 

\begin{assumption}
\label{regularity} (\textit{Score Regularity and Nature of Estimated Nuisance}) For all $n\ge 3$, let $\mathcal{T}_n$ be a set containing $\eta_0$ satisfying the following properties:
\begin{itemize}
    \item[(a)] Given a random subset $I$ (with $T := |I|$) of $[n]$ with $T/n \overset{a.s.}{\to}1/K$, the nuisance parameter $\hat\eta$ estimated from $\{W_i\}_{i\in I^c}$ is not contained in $\mathcal{T}_n$ only finitely often.
    \item[(b)] For $q>2$, the following moment conditions hold:  
    $$m_n:=\sup_{\eta\in \T_n}(\E\|\psi(W;\theta_0,\eta)^q\|)^{1/q}\le c_1; \ \ \ m_n':=\sup_{\eta\in \T_n}(\E\|\psi^a(W;\theta_0,\eta)^q\|)^{1/q}\le c_1.$$
    \item[(c)] The following rate conditions hold:
    \begin{align}
        r_n&= \sup_{\eta\in \T_n} \|\E[\psi^a(W;\eta)] - \E[\psi^a(W;\eta_0)]\|\le \delta_n. \label{rn}\\
        r_n'&= \sup_{\eta\in \T_n} (\E[\|\psi(W;\theta_0,\eta) - \psi(W;\theta_0,\eta_0)\|^2])^{\frac 12}\le \delta_n. \label{rn'}\\
        \lambda_n'&= \sup_{r\in (0,1);\eta\in \T_n} \left\|\dfrac{d^2}{dr^2}\E[\psi(W;\theta_0;\eta_0+r(\eta-\eta_0))]\right\|\le {\delta_n\sqrt{\dfrac{\log\log n}{n}}}. \label{lambdan'}
    \end{align}
    \item[(d)] The score is non-degenerate: all eigenvalues of the matrix $\E[\psi(W;\theta_0,\eta_0)\psi(W;\theta_0,\eta_0)^T]$ are at least $c_0$.
\end{itemize}
\end{assumption}

Assumption~\ref{regularity} (a) ensures that $\hat\eta$ is reasonably well-behaved, i.e., for a large enough sample size, $\hat\eta$ is always contained in $\mathcal{T}_n$. Since $\mathcal{T}_n\subseteq\mathcal{T}$ creates a shrinking neighborhood around $\eta_0$ at rates $r_n$, $r_n'$ and $\lambda_n'$ via \cref{rn,rn',lambdan'}, they ensure strong consistency properties of the estimator $\hat\eta$. The rate $\lambda_n'$ is of particularly interest, as it determines how fast the nuisance estimates need to converge to the true parameters with respect to the growing sample size. While in some problems this double derivative amounts to be 0, for instance in \eqref{plr1} when $m_0$ is known, or in an experimental study to obtain the ATE with known propensity scores $e_0$, in most cases that is not the case; and hence the double derivative condition imposes restrictions on the quality of the estimators used. However, it can be ensured that if $\psi$ happens to be an influence function or a pseudo-outcome emulating an influence function, one can ensure that $\lambda_n'$ is of the order $O(\|\hat\eta-\eta_0\|^2)$ (\cite{fisher2021visually}), and hence in the worst case $\eta_0$ needs to be strongly estimated at rate $(\log\log n/n)^\frac14$.

\textbf{Remark}: One can contrast the nature of Assumptions~\ref{neyman-orthogonal} and~\ref{regularity} to that of \cite[Assumptions 3.1 and 3.2]{chernozhukov2018double}, to see that they are remarkably similar, except that $\mathcal{T}_n$ is now a nuisance realization set that is allowed to not contain the nuisance estimate $\hat\eta$ only finitely often, and the rate conditions in \cref{orthogonal,rn,rn',lambdan'} allow an additional factor of $\sqrt{\log\log n}$ to compensate for the strong consistency expected of the nuisance realization set. This is unsurprising, as the law of iterated logarithm --- a key result underpinning several results in our paper --- states that a mean-centered random walk with finite variance is of $O(\sqrt{n\log\log n})$, whereas the central limit theorem can only ensure that the same is $O_p(\sqrt{n})$. 

\begin{restatable}{theorem}{samplesplitting}
\label{dml-linear} (\textit{Properties of the Sequential Sample Splitting Estimator}) Under Assumption~\ref{neyman-orthogonal} and~\ref{regularity},
$$\sigma^{-1}(\hat\theta-\theta_0) = -\dfrac 1T\sum_{i=1}^T \sigma^{-1}J_0^{-1}\psi(W_i;\theta_0,\eta_0) + O\left((r_n+r_n')\sqrt{\dfrac{\log\log n}{n}} + \lambda_n +\lambda_n'\right), $$
where $\sigma^2 = J_0^{-1}\P[\psi\psi^T]J_0^{-1}$. 
By Assumption~\ref{neyman-orthogonal}(d) and~\ref{regularity}(c), the error term is of the order $O\left(\delta_n \sqrt{\dfrac{\log\log n}{n}}\right)$, and thus $o\left( \sqrt{\dfrac{\log\log n}{n}}\right)$ if $\delta_n\to 0$.
\end{restatable}

Theorem~\ref{dml-linear} ensures that the sample splitting estimator $\hat\theta$ admits an asymptotically linear representation, that holds almost surely. The almost sure holding of the representation is a primary distinction from previous results, and thus allows one to use representation of random sums to create AsympCS for $\theta_0$ from $\hat\theta$. Of course, the variance $\sigma$ needs to be estimated with a strong consistency guarantee to ensure that such confidence sequences can be computed from the data, which brings us to Theorem~\ref{var-est}.

\begin{restatable}{theorem}{variance}
\label{var-est}
(\textit{Estimation of Variance}) Suppose Assumption~\ref{neyman-orthogonal} and~\ref{regularity} hold, 
then estimating $\sigma^2$ by $\hat\sigma^2 = \hat J_{0T}^{-1}\P_T[\hat\psi\hat\psi'](\hat J_{0T}^{-1})'$ with $\hat J_{0T} = \frac 1T\sum_{i=1}^T \psi_a(W_i,\hat\eta)$, we have
$$\hat\sigma^2 = \sigma^2 +  O\left(n^{-(1-\frac 2q)}\bigvee\sqrt{\dfrac{\log\log n}{n}} + r_n + r_n'\right).\ 
$$ 
\end{restatable}

Theorem~\ref{dml-linear} and Theorem~\ref{var-est} allows us to obtain a asymptotically linear representation of $\hat\sigma^{-1}(\theta - \theta_0)$ that holds almost surely, as illustrated in Corollary~\ref{combining_all}.

\begin{corollary}
\label{combining_all}
Under Assumption~\ref{neyman-orthogonal} and~\ref{regularity} 
$$\hat\sigma^{-1}(\hat\theta-\theta_0) = -\dfrac 1T\sum_{i=1}^T \sigma^{-1}J_0^{-1}\psi(W_i;\theta_0,\eta_0) + O\left(n^{-(1-\frac 2q)}\bigvee\sqrt{\dfrac{\log\log n}{n}} + (r_n+r_n')\sqrt{\dfrac{\log\log n}{n}} + \lambda_n +\lambda_n'\right).$$
In words, the $\sigma^{-1}$ on the left hand side of Theorem~\ref{dml-linear} can be replaced by $\hat\sigma^{-1}$ with an additional error of the order at most $O\left(n^{-(1-\frac 2q)}\bigvee\sqrt{\dfrac{\log\log n}{n}}\right)$.
\end{corollary}

The estimator in Theorem~\ref{dml-linear} have all been obtained by sample-splitting, which needless to say, may lose out on efficiency. An easy remedy to this is cross-fitting, which can even be adapted to the sequential setting by randomly assigning a newly arrived data point to one of the folds of the $K$-fold partition, or by looping through the folds sequentially. One can define two DML estimators based on such cross-fitting technique.

\begin{definition}
    \label{cross-fit}
    Consider a $K$ fold random partition $\{I_k\}_{k=1}^K$ of $[n]$ such that $|I_k|/n\overset{a.s.}{\to} 1/K$, and for each $k$, estimate $\hat\eta_k = \hat\eta(\{W_i\}_{i\in I_k^c}$. 
    \begin{itemize}
        \item ($\text{DML}_1$ estimator) Construct $\hat\theta_k$ as the solution to \begin{equation}\dfrac{1}{|I_k|}\sum_{i\in I_k}\psi(W_i,\theta,\hat\eta_k) = 0,\label{dml1}\end{equation} and the corresponding variance estimate as $$\hat\sigma^2_k = \hat J_{0k}^{-1}\left(\dfrac{1}{|I_k|}\sum_{i\in I_k}\psi(W,\theta,\hat\eta_k)\psi(W,\theta,\hat\eta_k)\right)(\hat J_{0k}^{-1})^T,$$ where $\hat J_{0k} = \frac{1}{|I_k|}\sum_{i\in I_k}\psi^a(W_i,\hat\eta_k)$. Then the $\dml_1$ estimator for $\theta$ is defined as $ \hat\theta_{\dml_1} = \dfrac{1}{K}\sum_{k=1}^K \hat\theta_k,$ with estimated variance $\hat\sigma^2_{\dml_1} = \frac{1}{K}\sum_{k=1}^K \hat\sigma^2_k$.

        \item ($\dml_2$ estimator) $\hat\theta_{\dml_2}$ is defined as the solution to 
        \begin{equation}\dfrac 1{K}\sum_{k=1}^K\dfrac{1}{|I_k|}\sum_{i\in I_k} \psi(W_i,\theta,\hat\eta_k) = 0,\label{dml2}\end{equation} and denoting $\hat J_0 = \frac 1K\sum_{k=1}^K \hat J_{0k}$, its variance estimate is given by $$\hat\sigma^2_{\dml_2} = \hat J_0^{-1}\left(\dfrac 1K\sum_{k=1}^K\dfrac{1}{|I_k|}\sum_{i\in I_k}\psi(W,\theta,\hat\eta_k)\psi(W,\theta,\hat\eta_k)\right)(\hat J_0^{-1})^T.$$
    \end{itemize}    
\end{definition}

The choice of $K$ has no asymptotic impact on the confidence sequences produced, but may produce better performance on moderate sample size. \cite{chernozhukov2018double}  recommends using $K=4$ or $5$ for better performance than $K=2$. Also $\dml_2$ is generally preferred than $\dml_1$ due to better stability in the pooled empirical Jacobian in \eqref{dml2} than its individual counterpart in \eqref{dml1}. Analogous results to Theorem~\ref{dml-linear} and Theorem~\ref{var-est} hold for the cross-fitting estimators in Definition~\ref{cross-fit}, the result of which we shall state without proof in Corollary~\ref{seqdml-all}. 

\begin{corollary}
\label{seqdml-all}
Under Assumptions~\ref{neyman-orthogonal} and~\ref{regularity}, both the $\dml_1$ and $\dml_2$ estimators satisfy 
\begin{align*} \hat\sigma_{\dml_j}(\hat\theta_{\dml_j} - \theta_0) &= -\dfrac 1n\sum_{i=1}^n\sigma^{-1}J_0^{-1}\psi(W_i;\theta_0,\eta_0) \\ & \hspace{4em}+ O\left(n^{-(1-\frac 2q)}\bigvee\sqrt{\dfrac{\log\log n}{n}} + (r_n+r_n')\sqrt{\dfrac{\log\log n}{n}} + \lambda_n +\lambda_n'\right);
\end{align*}
for $j=1,2$.
\end{corollary}

\textbf{Remark}: Corollary~\ref{combining_all} and~\ref{seqdml-all}, while resembling their counterpart results in \cite{chernozhukov2018double}, are distinct in two key aspects. Firstly, the equality in Corollary~\ref{seqdml-all} holds almost surely, with the error terms boasting almost sure validity in rates, while in contrast the usual errors in approximation as in \cite{chernozhukov2018double} are in probability. While it may seem like a technical difference, it is crucial to construct confidence sequences with asymptotic anytime-validity, as can be seen in Theorem~\ref{csdml}. Also, the rate of approximations has an additional $\sqrt{\log\log n}$ in the numerator on the right hand side of Corollary~\ref{seqdml-all}. This is an artefact of the orthogonality and regularity assumptions in Assumption~\ref{neyman-orthogonal} and~\ref{regularity}, as well as the law of iterated logarithm (LIL) as discussed before. Since the LIL is sharp, it is unlikely that the $\sqrt{\log\log n}$ term can be eradicated.

Based on the asymptotically linear expansion in Corollary~\ref{seqdml-all}, we can use the recipe of \cite{waudby2021time} to obtain confidence sequence and confidence regions for $\theta_0$. 

\begin{restatable}{theorem}{confidence}
\label{csdml}
Suppose Assumptions~\ref{neyman-orthogonal} and~\ref{regularity} hold. Then, for any constant $\rho>0$,
\begin{itemize}
\item[(a)] $ $ \vspace{-1.5em}
$$\forall n\ge 1, \left\{\theta\in\R^d: \|\hat\sigma^{-1}_{\dml_j}(\hat\theta_{\dml_j}-\theta)\|^2<\dfrac{2(n\rho^2+1)}{n^2\rho^2}\log\dfrac{(n\rho^2+1)^\frac d2}{\alpha}\right\} $$ constitutes a $(1-\alpha)$-Asymp CS region for $\theta_0$, for $j=1,2$, with approximation rate $$O\left(n^{-(1-\frac 2q)}\bigvee\sqrt{\dfrac{\log\log n}{n}} + (r_n+r_n')\sqrt{\dfrac{\log\log n}{n}} + \lambda_n +\lambda_n'\right).$$
\item[(b)] For any $l\in \R^d$ vector, a $(1-\alpha)$-Asymp CS for the scalar parameter $l^T\theta_0$ is given by $$l^T\hat\theta_{\dml_j} \pm \sqrt{l^T\hat\sigma^{-2}_{\dml_j}l}\sqrt{\dfrac{2n\rho^2+1}{n^2\rho^2}\log\dfrac{(n\rho^2+1)}{\alpha}}$$ for $j=1,2$, with the same approximation rate as in (a).
\end{itemize}
Also, replacing $\rho$ by $\rho_m: = \rho(\hat\sigma^2m \log (m\vee e)) = \sqrt{\dfrac{-2\log\alpha + \log(-2\log\alpha) + 1}{\hat\sigma_m^2 m\log(m\vee e)}}$, where $m$ denotes the first `peeking' time, provides asymptotic Type-I error coverage as in \eqref{type1} at level exactly equal to $\alpha$. 
\end{restatable}


Before illustrating the confidence sequences in Theorem~\ref{csdml}, it is important to note how easily the leap from regular confidence intervals to anytime-valid confidence sequences is achieved. With only mildly stronger conditions than those in \cite{chernozhukov2018double}, as outlined in Assumptions~\ref{neyman-orthogonal} and~\ref{regularity}, one can secure asymptotic anytime-valid inference and its associated benefits. Therefore, if one satisfies the conditions of \cite{chernozhukov2018double}, the additional effort for the anytime-valid guarantee is minimal and yields substantial advantages such as continuous monitoring and data-driven stopping.

\section{Inference in Partially Identified Observational Studies} \label{partial}

Consider the situation in Example~\ref{compulsory}. Under the standard assumption of unconfoundedness and other causal assumptions, one can identify and estimate the average treatment effect. More formally, consider a set of potential outcomes $Y(0), Y(1)$ and a treatment $A$, and the observed realization to be $Y = AY(1) + (1-A)Y(0)$, and a set of potential covariates (confounders) $X\in\mathcal{X}$ such that $0<\varepsilon<\P(A=1|X)<1-\varepsilon$. The standard no unmeasured confounding (NUC) assumption is given by $(Y(0),Y(1))\indep A|X$.
The augmented inverse probability weighting (AIPW) estimator is a common choice in estimating the ATE under NUC, which is known to have the Neyman-orthogonality property as in Definition~\ref{orthogonal-pure}, leading to double robust estimation. \cite{waudby2021time} themselves considers the anytime-valid inference using the AIPW estimator, and delineates conditions for the strong guarantee. However, their conditions are specifically catered to the AIPW estimator, and it is unclear how to extend the assumptions on more general estimators to obtain an asymptotically valid representation that holds almost surely. To illustrate this, consider the unconfoundedness assumption to be violated, which is often a potential concern in observational studies (\cite{rosenbaum2010design}). One often imagines an unmeasured confounder $U\in \mathcal{U}$ such that NUC may be generalized as:
\begin{assumption}[Unobserved-confounding]
    $(Y(0), Y(1))\indep A|X,U$. \label{uc}
\end{assumption}
However, as $U$ is unmeasured, as is often the case in most observational studies, one cannot point-identify the ATE without additional assumptions. The NUC assumption in an observational study, while acting as a currency for buying information (\cite{coombs1964theory}), may not necessarily be credible, and do not manifest testability from the observed data. In fact, \cite{bekerman2024planning} illustrate how an outcome may fail to remain significant even when accounting for a small deviation from NUC, and hence robustness of the ATE estimate for mild deviations from NUC is indispensable in observational studies for trustworthy inference. In such situations, one may be more comfortable providing bounds on the ATE in a milder assumption framework, which allows for a more credible inference. One such model is based on the ideas of sensitivity analysis of Rosenbaum (\cite{rosenbaum2002observational, rosenbaum2010design}) as in Assumption~\ref{rosenmodel}.

\begin{assumption}[$\Gamma$-selection bias model] \label{rosenmodel}
    With $1\le \Gamma<\infty$ \begin{align*}\dfrac{1}{\Gamma}\le \dfrac{\P(A=1|X,U=u)/\P(A=0|X,U=u')}{\P(A=0|X,U=u)/\P(A=1|X,U=u')}\le \Gamma\ \ \forall u, u'\in\mathcal{U} \text{ almost surely}. 
    \end{align*}
\end{assumption}

Assumption~\ref{rosenmodel} limits how much the likelihood of receiving treatment can differ among similar units, based on their measured characteristics, using the odds ratio and a sensitivity parameter $\Gamma$. Note that $\Gamma =1$ brings one back to NUC, while $\Gamma =\infty$ leaves the influence of $U$ unrestricted. Realistically, one can hope for obtaining meaningful inference for some fixed $\Gamma>1$ based on bounds that the model~\ref{rosenmodel} implies on the ATE. 

\cite{yadlowsky2018bounds} showed that under Assumption~\ref{uc} and~\ref{rosenmodel}, that $\E[Y(1)]$ can be bounded below sharply by $\mu_1^- = \E[AY+(1-A)g_1(X)]$, where $g_1(X)$ solves the optimization problem 
\begin{equation}
    \min_{g(\cdot)} \E[(Y-g(X))^2_+ + \Gamma(Y- g(X))^2_-|Z=1], \label{regbound}
\end{equation}
where $a_+ = \max\{a,0\}$ and $a_- = -\min\{a,0\}$. Analogous results hold for the upper bound on $\E[Y(0)]$ by $\mu_0+$ replacing $\Gamma$ by $\Gamma^{-1}$ and the conditioning event by $Z=0$ in \eqref{regbound}. $\mu^-_1 - \mu_0^+$ thus creates a lower bound on the ATE $\E[Y(1) - Y(0)]$. The upper bound on the ATE via $\mu_1^+$ and $\mu_0^-$ is similar. 

Before moving on to construction of AsympCS for the ATE, it must be emphasized how the utility of anytime-valid inference is all the more pronounced for any partial identification problem. To take this example, unless $\Gamma=1$ (which is situation of NUC and has been captured by the AIPW estimator), $\mu_1^- -\mu_0^+ < \mu_1^+ - \mu_0^-$, and the inequality being strict implies that the ATE can never be point identified, even with infinite data $(Y_i,A_i,X_i)_{i\ge 1}$. Thus the benefits of collecting more data hits a ceiling in such situations, contrary to traditional identifiable estimands. Now the gap between the upper and lower bound, although estimable, is not known prior to looking at the data, and hence fixing a sample size apriori for meaningful downstream inference is impractical without peeking at the data. The asymptotic anytime-valid framework thus allows one to stop expensive data collection as soon as one obtains ``good enough" inference (eg: bounds are tight enough, bounds seem to have converged, confidence sequence does not contain 0, etc.) and is hence especially catered in efficient handling of such situations, even in observational settings. 

For the rest of this section we focus on estimating $\mu_1^-$, as the remaining parameters are estimated similarly. Let $W = (Y,A,X)$, and consider the functional 
\begin{equation}
\psi(W; \mu_1^-, \eta) = AY + (1-A)g_1(X) + A\dfrac{(Y-g_1(X))_+ -\Gamma(Y-g_1(X))_-}{\nu(X)}\dfrac{1-e(X)}{e(X)} - \mu^1_- \label{patefunctional}
\end{equation}
as demonsrated by \cite{yadlowsky2018bounds}, where $\eta = (g_1,e,\nu)$ is the set of nuisance parameters, where $g_1$ is as described in \eqref{regbound}, $e(X) = \P(A=1|X)$ is the propensity score, and $\nu$ is an additional nuisance parameter given by
\begin{align*}
    \nu(X) = \P(Y\ge g_1(X)|A=1,X) + \Gamma\P(Y<g_1(X)|A=1,X).
\end{align*}
We consider a split-sampling framework similar to what we discussed in Section~\ref{main}, where we estimate the nuisance parameters from one part of the sample, and equate the empirical expectation of $\psi$ from \eqref{patefunctional} in the remaining sample with plugin estimates for $\eta$ to $0$, to obtain $\hat\mu_1^-$. We assume the following strong consistency conditions on $\hat\mu^1_-$ and the underlying law for $W$:

\begin{assumption}
    \label{partial-ass}
    \begin{itemize}
        \item[(a)] (Regularity Condition) $\E[Y(1)^q] \le C_q$ for some $q>2$, and $Y(1)$ admits a uniformly bounded conditional density (given $X=x$ and $Z=1$) with respect to the Lebesgue measure. 
        \item[(b)] (Boundedness Condition) $\exists \varepsilon>0$ such that $e_1(X)\in [\varepsilon, 1-\varepsilon]$ a.s., $\hat e_1(X)\in [\varepsilon, 1-\varepsilon]$ for all but finitely many $n$. 
        \item[(c)] (Strong Consistency Condition) $\|\hat g_1(\cdot) - g_1(\cdot)\|_{1} \overset{a.s.}{\to} 0$, $\hat\eta_1\to \eta_1$ almost surely and $\|\hat\eta_1(\cdot) - \eta_1(\cdot)\|_{\infty} = O(1)$.
        \item[(d)] (Rate Condition) $\|\hat\eta_1 - \eta_1\|_2 = o\left(\left[\dfrac{\log\log n}{n}\right]^{\frac 14}\right)$.
    \end{itemize}
\end{assumption}


Assumption~\ref{partial-ass}(a) is a regularity condition on the underlying distribution of the potential outcomes, while (b) is a strengthening of the positivity condition to ensure the propensity scores are mildly well-behaved. (c) and (d) on the other hand impose conditions on the estimator being used for the nuisance parameters. One can view the estimation of the parameter $g$ as an M-estimation problem based on \eqref{regbound}, strong consistency properties of which has been considered in \cite{schreuder2020nonasymptotic} in a low-dimensional setting. In a high-dimensional setting, common estimation choices are sieve-based regression (\cite{geman1982nonparametric}) or kernel-based methods (\cite{devroye1981almost}), the almost sure properties of which have been well studied in the special case of $\Gamma = 1$ (\cite{lugosi1995nonparametric, gordon1984almost}).  Strong consistency rates for such nonparametric regressions have also been considered in (\cite{chen1990extension}), which prove a LIL rate for nonparametric regressions. On the other hand, $\hat e$ is essentially a propensity score model, and can be estimated at a LIL rate based on \cite{yang2019law} under a correctly specified GLM model. Finally, estimation for $\nu$ can also be treated as an M-estimation problem, or can be considered as problems pertaining to functionals of the conditional distribution of $Y|X,Z=1$, and hence allow the use of smoothed Nadaraya-Watson estimators or rank nearest neighbors, the LIL properties of which have been established in \cite{mehra2005pointwise}.

\begin{restatable}{theorem}{partialthm}
    \label{partial-theorem}
    Suppose the conditions of Assumption~\ref{partial-ass} hold, and $\hat\sigma^2_{1-} = \frac 1n\sum_{k=1}^K\sum_{i\in I_k} \psi(W;\hat\mu_1^-,\hat\eta)^2$. Then for any $\rho>0$,
    $$\hat\mu_1^- \pm \hat\sigma_{1-}\sqrt{\dfrac{2n\rho^2+1}{n^2\rho^2}\log\dfrac{(n\rho^2+1)}{\alpha}}$$ constitutes a $(1-\alpha)$ AsympCS for $\mu_1^-$.
\end{restatable}

Analogous results also hold for the remaining parameters $\mu_0^+$, $\mu_1^+$ and $\mu_0^-$. Theorem~\ref{partial-theorem} thus allows us to construct confidence sequences for the partially identified ATE, and the proof of the theorem, as can be seen from Appendix~\ref{partialproof}, follows just by verifying the conditions of Assumptions~\ref{neyman-orthogonal} and~\ref{regularity}.

\section{Inference in an Instrumental Variable based Online Setting} \label{iv-late}

We turn our attention to situations in Example~\ref{coupon} or~\ref{drug}. In both these situations, while experimental allocation of treatments has been randomized, uptake of treatments is a personal choice for the units, and hence they may not comply with the treatment assigned to them. Such uptake of treatment may be associated with the outcome of interest (someone may use a coupon only if they consider it worth to themselves, or only health-seeking people adhere to the uptake of the new drug), and hence does not allow identification of causal estimands without further assumption. Under such situations, an important technique for estimating a causal effect is the use of instrumental variables, as alluded to in Section~\ref{introduction}. Formally, let $Z\in\mathcal{Z}$ be a variable (here treatment status assigned by the experimenter), $A$ the true treatment status, and let $Y$ denote the realized outcome. Let $U$ be a potential unmeasured confounder that may affect both $Y$ and $A$, and $X$ as before denote measured covariates. 
\begin{assumption}[IV-assumptions] \label{stdivass}
For $Z$ to be a valid IV, one must have 
(i) Exclusion restriction: $Z\indep Y|(A,U,X)$. (ii) Independence: $Z\indep U|X$. (iii) $Z\not\indep A|X$.
\end{assumption}

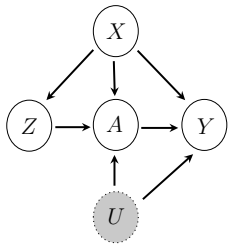
\begin{figure}[!ht]
\centering
\tikzset{every picture/.style={line width=0.75pt}} 

\scalebox{0.5}{
\begin{tikzpicture}[x=0.5pt,y=0.5pt,yscale=-1,xscale=1]

\draw   (17,206.77) .. controls (17,185.44) and (32.22,168.15) .. (51,168.15) .. controls (69.78,168.15) and (85,185.44) .. (85,206.77) .. controls (85,228.11) and (69.78,245.4) .. (51,245.4) .. controls (32.22,245.4) and (17,228.11) .. (17,206.77) -- cycle ;
\draw   (148,205.64) .. controls (148,184.3) and (163.22,167.01) .. (182,167.01) .. controls (200.78,167.01) and (216,184.3) .. (216,205.64) .. controls (216,226.97) and (200.78,244.26) .. (182,244.26) .. controls (163.22,244.26) and (148,226.97) .. (148,205.64) -- cycle ;
\draw  [fill={rgb, 255:red, 155; green, 155; blue, 155 }  ,fill opacity=0.54 ][dash pattern={on 0.84pt off 2.51pt}] (148,345.37) .. controls (148,324.04) and (163.22,306.75) .. (182,306.75) .. controls (200.78,306.75) and (216,324.04) .. (216,345.37) .. controls (216,366.71) and (200.78,384) .. (182,384) .. controls (163.22,384) and (148,366.71) .. (148,345.37) -- cycle ;
\draw   (282,205.64) .. controls (282,184.3) and (297.22,167.01) .. (316,167.01) .. controls (334.78,167.01) and (350,184.3) .. (350,205.64) .. controls (350,226.97) and (334.78,244.26) .. (316,244.26) .. controls (297.22,244.26) and (282,226.97) .. (282,205.64) -- cycle ;
\draw   (148,63.63) .. controls (148,42.29) and (163.22,25) .. (182,25) .. controls (200.78,25) and (216,42.29) .. (216,63.63) .. controls (216,84.96) and (200.78,102.25) .. (182,102.25) .. controls (163.22,102.25) and (148,84.96) .. (148,63.63) -- cycle ;
\draw [line width=1.5]    (90.5,209.04) -- (140,209.04) ;
\draw [shift={(144,209.04)}, rotate = 180] [fill={rgb, 255:red, 0; green, 0; blue, 0 }  ][line width=0.08]  [draw opacity=0] (13.4,-6.43) -- (0,0) -- (13.4,6.44) -- (8.9,0) -- cycle    ;
\draw [line width=1.5]    (220.5,210.18) -- (272.5,210.18) ;
\draw [shift={(276.5,210.18)}, rotate = 180] [fill={rgb, 255:red, 0; green, 0; blue, 0 }  ][line width=0.08]  [draw opacity=0] (13.4,-6.43) -- (0,0) -- (13.4,6.44) -- (8.9,0) -- cycle    ;
\draw [line width=1.5]    (180,297.66) -- (180,255.08) ;
\draw [shift={(180,251.08)}, rotate = 90] [fill={rgb, 255:red, 0; green, 0; blue, 0 }  ][line width=0.08]  [draw opacity=0] (13.4,-6.43) -- (0,0) -- (13.4,6.44) -- (8.9,0) -- cycle    ;
\draw [line width=1.5]    (179,109.07) -- (180.39,158) ;
\draw [shift={(180.5,162)}, rotate = 268.38] [fill={rgb, 255:red, 0; green, 0; blue, 0 }  ][line width=0.08]  [draw opacity=0] (13.4,-6.43) -- (0,0) -- (13.4,6.44) -- (8.9,0) -- cycle    ;
\draw [line width=1.5]    (223,318.11) -- (296.02,252.61) ;
\draw [shift={(299,249.94)}, rotate = 138.11] [fill={rgb, 255:red, 0; green, 0; blue, 0 }  ][line width=0.08]  [draw opacity=0] (13.4,-6.43) -- (0,0) -- (13.4,6.44) -- (8.9,0) -- cycle    ;
\draw [line width=1.5]    (215,93.16) -- (284.23,165.26) ;
\draw [shift={(287,168.15)}, rotate = 226.16] [fill={rgb, 255:red, 0; green, 0; blue, 0 }  ][line width=0.08]  [draw opacity=0] (13.4,-6.43) -- (0,0) -- (13.4,6.44) -- (8.9,0) -- cycle    ;
\draw [line width=1.5]    (148,92.03) -- (79.73,165.22) ;
\draw [shift={(77,168.15)}, rotate = 313.01] [fill={rgb, 255:red, 0; green, 0; blue, 0 }  ][line width=0.08]  [draw opacity=0] (13.4,-6.43) -- (0,0) -- (13.4,6.44) -- (8.9,0) -- cycle    ;

\draw (165,49) node [anchor=north west][inner sep=0.75pt]  [font=\LARGE] [align=left] {$\displaystyle X$};
\draw (165,190) node [anchor=north west][inner sep=0.75pt]  [font=\LARGE] [align=left] {$\displaystyle A$};
\draw (36,193) node [anchor=north west][inner sep=0.75pt]  [font=\LARGE] [align=left] {$\displaystyle Z$};
\draw (303,193) node [anchor=north west][inner sep=0.75pt]  [font=\LARGE] [align=left] {$\displaystyle Y$};
\draw (167,330) node [anchor=north west][inner sep=0.75pt]  [font=\LARGE] [align=left] {$\displaystyle U$};

\end{tikzpicture}}
\label{iv-dag}
\caption{DAG to represent IV Assumption~\ref{stdivass}}
\end{figure}

Assumption~\ref{stdivass}, as encoded in Figure~\ref{iv-dag}, allows us to define the potential outcomes $Y(a)$ based on the level of $A=a$, and $Y = Y(A)$ by the exclusion restriction (\cite{wang2018bounded}). In the context of Example~\ref{coupon} and~\ref{drug}, for instance, the group assignment the units were assigned to (ie coupon vs no coupon, treatment vs placebo), may be assumed to be a valid instrument satisfying Assumption~\ref{stdivass}. However, Assumption~\ref{stdivass} is not adequate to nonparametrically identify a causal estimand. The classic way to deal with it is the monotonicity assumption (\cite{angrist1996identification}), which in the case of binary treatment and binary instrument can be encoded by
\begin{assumption}[IV-monotonicity] \label{monoton}
    $A(Z=1)\ge A(Z=0)$.
\end{assumption}
Assumption~\ref{monoton} essentially means that there are no defiers, i.e., for no unit is it the case that $A = 1$ when $Z= 0$. In Example~\ref{coupon} or~\ref{drug}, Assumption~\ref{monoton} holds by design, as no unit can enroll itself into treatment if not assigned treatment. Thus, since non-compliance can occur only in a particular way, one can identify a meaningful causal estimand.

Under Assumption~\ref{uc},~\ref{stdivass} and~\ref{monoton}, one can identify the local average treatment effect (LATE) given by $\theta_0 = \E[Y(1) - Y(0)| A(1)> A(0)]$, that is, the average treatment effect on the units that comply with the treatment (ie, those who take treatment if assigned so, and not if not), via the Wald estimand $$\dfrac{\E[Y|Z=1] - \E[Y|Z=0]}{\E[A|Z=1] - \E[A|Z=0]} = \dfrac{\E[\E[Y|Z=1,X] - \E[Y|Z=0,X]]}{\E[\E[A|Z=1,X] - \E[A|Z=0,X]]}.$$ 

A business that is concerned about the expenses of running an experiment, or a drug trial concerned about medical ethics, may be interested to evaluate the effect of treatment on the outcome of interest, and hence would want to stop the experiment as soon as possible: to deploy the treatment if it's beneficial conclusively, and to stop if not. In such settings, one may be interested in deploying anytime-valid inference to obtain confidence sequences for the LATE.

Consider the data $\{W_i\}_{i\ge 1}$ iid, with $W_i = (Y_i,A_i,Z_i,X_i)$, nuisance parameters $\eta = (g_t,g_c,m_t,m_c, e)$ and the functional based on \cite{tan2006regression}:
\begin{equation}
    \begin{aligned}
   \psi(W; \theta,\eta) &= g_t(X) - g_c(X) + \dfrac{Z(Y-g_t(X))}{e(X)} - \dfrac{(1-Z)(Y-g_c(X))}{1-e(X)} \\
   &- \theta \left[ m_t(X) - m_c(X) + \dfrac{Z(A-m_t(X))}{e(X)} - \dfrac{(1-Z)(A-m_c(X))}{1-e(X)}\right]. 
   \end{aligned}
   \label{iv-functional}
\end{equation}

Here $g_t(X) = \E[Y|Z=1,X]$, $g_c(X) = \E[Y|Z=0,X]$, $m_t(X) = \E[A|Z=1,X]$, $m_c(X) = \E[A|Z=0,X]$, $e(X) = \E[Z|X]$ are the relevant nuisance parameters, which is estimated using ML methods from part of the sample as discussed in Section~\ref{main}, and $\theta$, the parameter of interest that identifies the LATE, is obtained by equating the empirical expectation of $\psi$ from \eqref{iv-functional} on the remaining sample with plugin estimates for $\eta$, to zero. We assume the following conditions on the estimator for nuisance parameters $\hat\eta$, and the underlying law governing $W$.

\begin{assumption}
    \label{iv-ass}
    \begin{itemize}
        \item[(a)] (Regularity Condition) For $a=0,1$, $\E[Y(a)^q]\le C_q$ for some $q\ge 4$.
        \item[(b)] (Boundedness Condition) $\exists \varepsilon>0$ such that $e_0(X)\in [\varepsilon, 1-\varepsilon]$ a.s., $\|g_{a,0}\|_q = O(1)$, $a=0,1$.
        \item[(c)] (Identification Condition) $\E[m_{t,0}(X) - m_{c,0}(X)]\ge c$.
        \item[(d)] (Strong-consistency Condition) $\|\hat\eta(\cdot) - \eta_0(\cdot)\|_2\overset{a.s.}{\to} 0$, $\hat\eta\to\eta_0$ almost surely and $\|\hat\eta(\cdot) - \eta_0(\cdot)\|_\infty = O(1)$.
        \item [(e)] (Rate Condition) $$\|\hat e- e_0\|_2\left(\sum_{a\in \{t,c\}} \|\hat g_a(\cdot) - g_{a,0}(\cdot)\|_2 + \sum_{a\in \{t,c\}} \|\hat m_a(\cdot) - m_{a,0}(\cdot)\|_2\right) = o\left(\left[\dfrac{\log\log n}{n}\right]^\frac12\right).$$
    \end{itemize}
\end{assumption}

Assumption~\ref{iv-ass} (a) is a regularity condition ensuring that the counterfactual laws are well behaved, while (b) and (c) solidifies conditions on positivity and identifiability based on the propensity score and the IV-strength respectively. On the other hand, (d) and (e) imposes restrictions on the estimators being used. However, $m_t,m_c, g_t$ and $g_c$ are all real valued regression estimators, and as discussed in the discourse ensuing Assumption~\ref{partial-ass}, such parameters can be strongly estimated at LIL rates using kernel based methods or partitioning techniques. One may also use regression models for estimating the propensity scores $\hat e$, or use a GLM specification as has been considered in \cite{yang2019law}. Note that, in contrast with Assumption~\ref{partial-ass}, where the $L_2$ error in each of the parameters needed to be reasonably small (although not at the parametric rate), here the error condition is of the form of a mixed bias, and thus the error in estimation in the outcome regression functions can be compensated by that in the propensity score. 

\begin{restatable}{theorem}{iv}
    \label{iv-theorem}
    Suppose the conditions of Assumptions ~\ref{iv-ass} are met. Then, with $\hat\sigma^2$ as in Theorem~\ref{partial-theorem}, and for any $\rho>0$, $$\hat\theta\pm \hat\sigma\sqrt{\dfrac{2n\rho^2+1}{n^2\rho^2}\log\dfrac{n\rho^2+1}{\alpha}}$$ constitutes an $(1-\alpha)$AsympCS for $\theta_0$.
\end{restatable}

Again as before, Theorem~\ref{iv-theorem} is a simple consequence that may be derived easily by verification of the conditions of Assumptions~\ref{neyman-orthogonal} and~\ref{regularity}, as is evident from the proof in Appendix~\ref{lateproof}.


\section{Experiments} \label{empirical}
In the following, we evaluate our approach on synthetic data, and illustrate its utility in real-world applications.

\subsection{Synthetic Data Experiments} \label{sims}

\subsubsection{Partial Identification with Unmeasured Confounding}

In this subsection, we demonstrate our findings applying to the formulations in Section \ref{partial}. We follow a data generation setup similar to that in \cite{yadlowsky2018bounds}, where we generate 
$X\sim \text{Uniform}[0,1]^d$, $U|X\sim \mathcal{N}\left(0, (1+\frac 12\sin (2.5 X_1))^2\right)$ and $Y(0) = \beta^TX + 5U, \ \ Y(1) = Y(0) +\tau$, where $d$ is the number of dimensions and $\tau$ the treatment effect. The treatment assignment $A$ is drawn i.i.d. Bernoulli with $$\P(A = 1|U,X) = \dfrac{\exp(\alpha_0 + X^T\mu + \log(\Gamma_\text{data})\mathbbm{1}(U>0))}{1 + \exp(\alpha_0 + X^T\mu + \log(\Gamma_\text{data})\mathbbm{1}(U>0))},$$ such that $$\dfrac{\P(A=1|X,U=u)/\P(A=0|X,U=u')}{\P(A=0|X,U=u)/\P(A=1|X,U=u')} = \Gamma_\text{data}^{\mathbbm{1}(u>0) - \mathbbm{1}(u'>0)}\in \left[\frac 1{\Gamma_\text{data}}, \Gamma_\text{data}\right]$$ satisfying Assumption \ref{rosenmodel} with parameter $\Gamma_\text{data}$. The $\alpha_0$ controls the treatment assignment ratio. The confidence sequences produced via applying Theorem \ref{partial-theorem} hold true as long as $\Gamma\ge \Gamma_\text{data}$, and, in the subsequent analysis, we set them both equal to $\exp(0.6)$. $\tau$ is set to be $-0.5$, and the hyperparameters are set as $\alpha_0 = 0$, $\mu\sim \mathcal{N}(0,I_d)$, $\beta\sim\mathcal{N}(0,I_d)$. We use a 4-fold cross-validation to generate the results, with an \texttt{xgboost} implementation to learn $\hat e, \hat g_1$ and $\hat\nu$, where $\hat g$ uses the special loss function $\eqref{regbound}$. $d$ is set to 4, and the experiments are run with sample sizes from 1000 to 50000 with an increment of 1000. The $\rho$-parameter in Theorem \ref{partial-theorem} is chosen to optimize for Robbins' normal mixture at the maximum allowed sample size, as delineated in \cite{waudby2021time}. 

\begin{figure}[!ht]
    \centering
    \includegraphics[width = \textwidth]{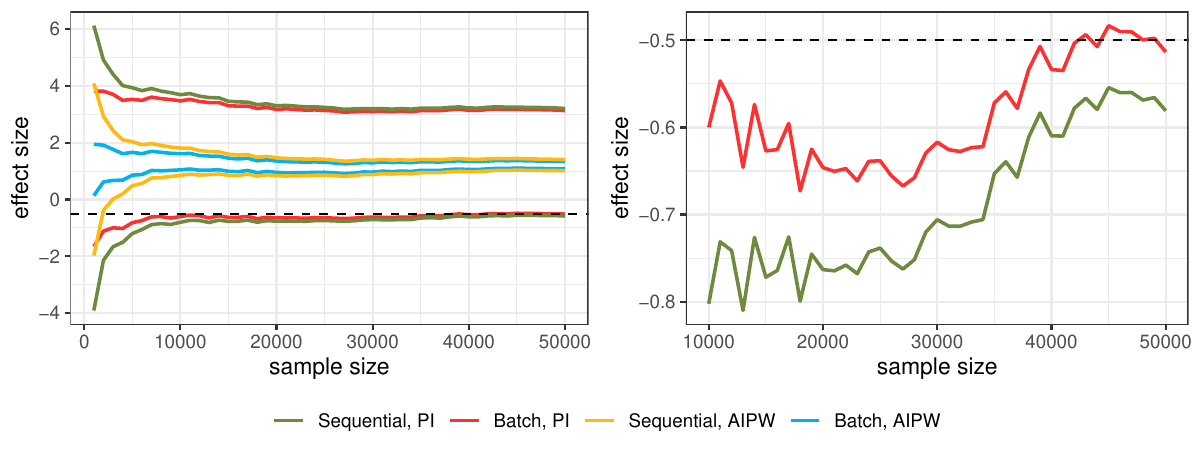}
    \caption{(a) The sequential and batch intervals on the partial identification (PI) AsympCS and CI, and their counterpart AIPW estimator that does not account for unmeasured confounding. (b) A zoomed in view of the sequential and batch lower confidence lines for the partial identification setup. The dashed line indicates the true treatment effect of -0.5.}
    \label{piplot}
\end{figure}

Figure \ref{piplot}(a) demonstrates the importance of PI-confidence regions to account for unmeasured confounding, since without it, even the AIPW estimator would infer an incorrect sign for the ATE. For an anytime-valid guarantee, however, just considering partial identification is not enough, as one also needs to account for protection against continuous monitoring. Figure \ref{piplot}(b) highlights the issue, zooming in on the lower ends of the AsympCS and the usual CI. Clearly, the true effect is not contained by the usual CIs uniformly over all sample sizes, while AsympCS holds up to its theoretical guarantees by always containing the true ATE.

\subsubsection{LATE estimation via Instrumental Variables}
Next, we illustrate the LATE estimation in a sequential setup. We generate $X\sim\mathcal{N}(0,I_d)$, and $U|X\sim \mathcal{N}\left(0, (0.5 + \sin X_1)^2\right)$. We generate a randomized experiment setup akin to \cite{farbmacher2022instrument}, where the instrument $Z\sim \text{Bern}(p)$, with $A = \mathbbm{1}(\alpha_z Z +U>0)$ guaranteeing the monotonicity condition as well as the IV-relevance condition on $A$ which is controlled via $\alpha_Z>0$. The outcome is generated as $Y = \theta A + \cos(U)(\beta^TX + U)$. We use $p = 0.4$ and $d=2$, and use an $\texttt{xgboost}$ implementation to learn the relevance condition. The hyperparameters are generated as $\beta\sim\mathcal{N}(0, 0.5 I_d)$, $\alpha_z = 2$, $p=0.4$, and the effect size $\theta = 3$, with a 4-fold cross-fitting to be used for estimation.

\begin{figure}[!ht]
    \centering
    \includegraphics[width = \textwidth]{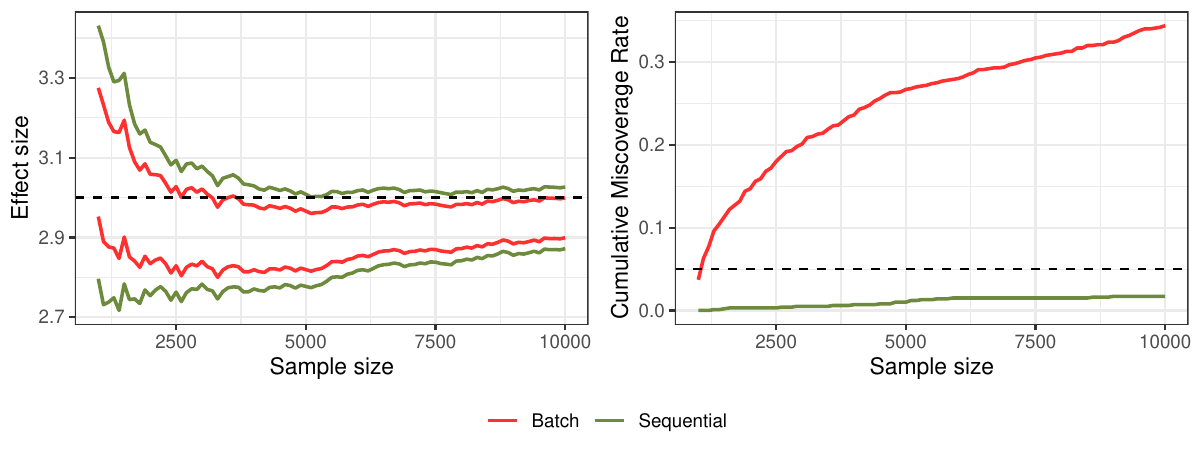}
    \caption{(a) The sequential and batch intervals on the estimation for LATE for one run of the experiment, with dashed line indicating a true LATE value of 3.0. (b) The cumulative miscoverage rate for 1000 repetitions of the experiment, with the dashed line showing the coverage level of 0.05.}
    \label{lateplot}
\end{figure}

Figure \ref{lateplot}(a) demonstrates a typical miscoverage by the batch approach that can occur for the LATE when not controlled for continuous monitoring, where we see that the upper bound is below the effect, while the AsympCS correctly accounts for it. Figure \ref{lateplot}(b) shows how the cumulative miscoverage rate grows rapidly for batch confidence intervals as the sample size increases, while the AsympCS controls the rate at the nominal level.

\subsection{Real-world Applications}

Empowered by the advantages showcased in ensuring the right confidence sequence to estimate the treatment effect, we deploy the methods in real-data instances.

\subsubsection{Impact of Academic Support Services on College Achievement}\label{star}

In our first application, we consider the impact that academic support services provide on college achievement. We use the data from \cite{angriststar} that evaluates the impact of the STAR program (Student Achievement and Retention Project), which had been designed to improve academic performance among college freshmen. Incoming first-year undergraduates were randomly assigned to one of three treatment arms: SSP (Student Support Program that offered peer-advising and supplemental instruction), SFP (Student Fellowship Program being awarded the opportunity to gain merit-based fellowships), and SFSP (who were given both SFP and SSP opportunities). Not everyone assigned treatment ended up accepting it, and hence intention-to-treat $Z$ does not necessarily match treatment status $A$. The monotonicity Assumption \ref{monoton} is automatically satisfied, and randomization of the treatment status makes it plausible that Assumption \ref{stdivass} is satisfied. The group assigned treatment (one of three treatment arms) consisted of 299 individuals out of 1255 in the entire dataset, of which 256 ended up taking treatment. We consider the sequential experiment in the order the individuals appear in the dataset, and also consider a sub-group analysis based on the sex of the individuals. We consider two outcomes: fall grades and GPA after the first year, in line with the original paper. The results are plotted in Figure \ref{star_plot}. The analysis, in line with the original work, fails to find a significant result (ie, the confidence sequence contains zero throughout) in almost all categories, except for Fall Grades for females. Using a sequential approach, one could have thus terminated the experiment earlier (since the confidence sequences are relatively stable), or even could have recruited more female participants in the study to gain greater confidence about the impact of the program on their fall grades. However, this application also highlights that sequential inference comes at a cost, especially for moderate sample sizes as in this dataset. The sequential method fails to obtain significance even for female Fall Grades, which highlights the trade-off between continuous peeking into the data and getting significance at the maximum sample size.

\begin{figure}[!ht]
    \centering
    \includegraphics[width=\linewidth]{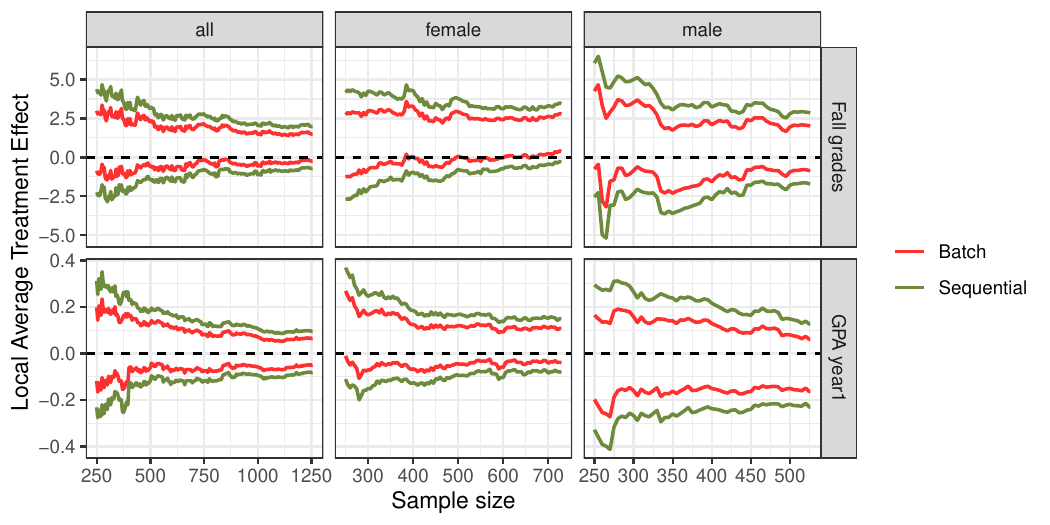}
    \caption{Local Average Treatment Effect for the STAR program for all students, and a gender-wise sub-group analysis.}
    \label{star_plot}
\end{figure}

\subsubsection{Differentially Expressed Genes in Leukemia Cells from Microarray Analysis in Observational Study} \label{gene}

In our second application, we apply anytime-valid inference to obtain differentially expressed genes in leukemia cells from microarray analysis. The causal inference approach of obtaining differentially expressed genes after accounting for potential confounders of age, sex, etc. has been expressed in the work of \cite{hellergene}, but they did not account for potential unmeasured confounding. Sequential inference is particularly useful here, as gene-profiling is expensive, and hence it is important to reach inference as soon as possible. We use the dataset from \cite{chen2022single}, which consists of gene profiling of 18 infants suffering from B-cell acute lymphoblastic Leukemia (B-ALL). We consider the B-cells as treatment and T-cells as control, and aim to obtain the difference in expressions for different genes across B and T cells. Since disease cells cannot be randomized, it is important to account for potential unmeasured confounding - a question that has been left open by \cite{hellergene}. The dataset consists of 20,043 genes, out of which 150 are reported to be significant by \cite{chen2022single}. We report confidence intervals of the gene-expression difference of these 150 genes after accounting for unmeasured confounding under Assumption \ref{rosenmodel} at $\Gamma = 1.5$, to account for small to moderate levels of potential unmeasured confounding. 76 genes turn out to remain significant at least once after using a burn-in of 5 donors or 33,554 cells, thus illustrating that adding a potential for unmeasured confounding reduces the number of significant genes by half. The sequential approach also illustrates how the experiment could have been stopped earlier --- setting a stopping time as the first time a gene is found to be significant, we obtain that the stopping time for different significant genes varies between 6 to 15 donors with a median of 6 donors, and thus collecting data for as many as 18 patients was relatively unnecessary. Figure \ref{diff-genes} illustrates two genes: one which remains significant at $\Gamma = 1.5$, and one that fails to be resilient to the same.

\begin{figure}[!ht]
    \centering
    \includegraphics[width=\linewidth]{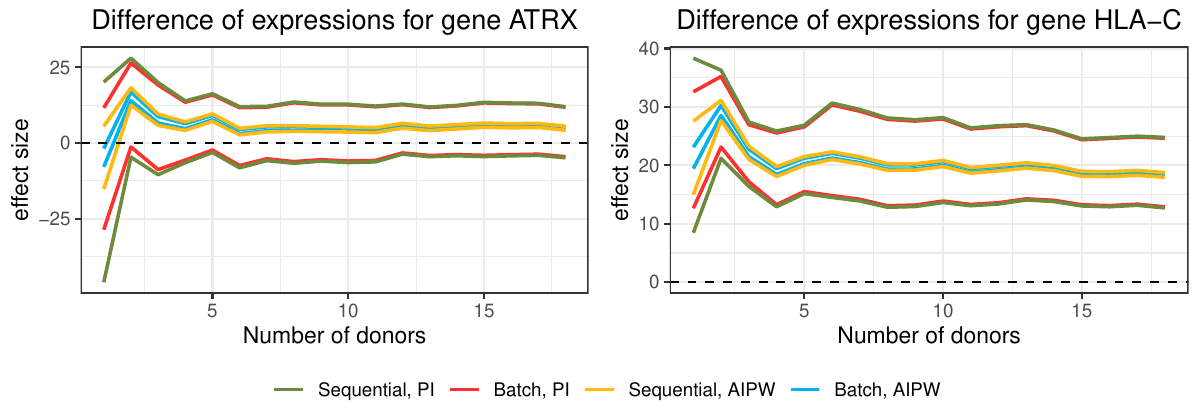}
    \caption{Average treatment effect of gene expression between T and B cells in B-ALL leukemia infant patients.}
    \label{diff-genes}
\end{figure}


\section{Discussion} \label{discuss}

In this article, we present a method to implement the double machine learning (DML) framework for estimating structural parameters in a sequential setting. Borrowing the concept of an asymptotic confidence sequence from \cite{waudby2021time}, we derive the precise conditions sufficient to ensure the sequential validity of DML-based intervals, and hence allowing to stop collecting data when inference is `good-enough'. We also provide a method for producing Neyman-orthogonal functions from an identifying equation. Our conditions for valid asymptotic confidence sequences in DML estimators are natural extensions of those in \cite{chernozhukov2018double}. 

Constructing DML estimators can be challenging, but our method enhances DML by slightly modifying the confidence intervals for stronger validity and flexibility, with minimal additional assumptions. Beyond point identification, we particularly emphasize its utility in partial identification, where more data would not narrow down confidence intervals to zero. Our approach allows researchers to stop data collection when results are `good enough' while maintaining coverage guarantees, making causal inference faster and more cost-effective across various scientific disciplines.

\subsection*{Acknowledgements}
Part of this work was developed during A. Dalal's Summer 2023 internship at Amazon Web Services. The authors thank William Bekerman and Anirban Chatterjee for helpful discussions.

\bibliography{ref.bib}

\newpage
\appendix
\section{Proof of Theorem~\ref{pseudooutcome}} \label{pseudoproof}

\first*
\begin{proof}
For simplicity, we shall use the notation  $\partial r f := \partial f/\partial r|_{r=0}$. From the identification condition, we know that for 
$$\E[\phi(W;\theta_0,\beta)|X] = 0.$$ Taking a regular paramteric sub-model parametrized by the index $r\in [0,1)$ and denoting $\beta_r = \beta_0 + r(\beta-\beta_0)$, we have, 
$$\E_r[\phi(W;\theta_0,\beta_r)|X] = 0.$$
Taking the derivative with respect to $r$ at $r=0$, we obtain, 
\begin{align*}
    0 &= \partial_r\E_r[\phi(W;\theta_0,\beta_r)|X] = \E_0[\phi(W;\theta_0,\beta_0)S(W|X)|X] + \partial_r\E[\phi(W;\theta_0, \beta_0 + r(\beta - \beta_0)|X] \\
    &= \E_0[\phi(W;\theta_0,\beta_0)S(W|X)|X] + \E_0[R(W;\theta_0,\eta_0)S(W|X)|X] \\
    &= \E_0[\psi(W;\theta_0,\eta_0)S(W|X)|X],
\end{align*}
where the third equality follows from \eqref{pseudocorrection}. However, as $R(W;\theta_0,\eta)$ has a mean of 0, we also have that 
for any regular parametric submodel, $$\E_r[\psi(W;\theta_0,\eta_0 + r(\eta - \eta_0)] = 0.$$
Denoting $\eta_r$ analogously to $\beta_r$ as above, and taking pathwise derivative, wrt to $r$ at $r=0$, we obtain, 
\begin{align*}
    0&= \partial_r\E_r[\psi(W;\theta_0,\eta_r)|X] = \E[\psi(W;\theta_0,\eta_0)S(W|X)|X] + \E[\partial_r \psi(W;\theta_0,\eta_r)|X].
\end{align*}
Since the first term in the above equation is $0$ as established, we have, $$\dfrac{\partial}{\partial r}\E[\psi(W;\theta_0,\eta + r(\eta - \eta_0)|X] = 0.$$ as desired. 
\end{proof}

\section{Proof of Results in Section~\ref{main}}
\label{mainproofs}

\subsection{Proof of Theorem~\ref{dml-linear}}
\textit{Notation}: We shall sometimes use $\hat\psi$ to denote $\psi(\cdot; \theta_0,\hat\eta)$ and $\psi$ for $\psi(\cdot;\theta_0,\eta_0)$ when clear from context.

\begin{lemma}[Lemma A.4 of \cite{waudby2021time}]
\label{decomposition}  (\textit{Decomposition of $\psi$}) 
$$\P_T\hat\psi - \P\psi = \Gamma_n^{SA} + \Gamma_n^{EP} + \Gamma_n^B$$
where 
\begin{align*}
    \Gamma_n^{SA} &:= (\P_T-\P)\psi & \text{is the sample average term, } \\
    \Gamma_n^{EP} &:= (\P_T-\P)(\hat\psi - \psi) & \text{is the empirical process term, and } \\
    \Gamma_n^B &:= \P(\hat\psi -\psi) & \text{is the bias term.}
\end{align*}
\end{lemma}



\begin{lemma}
\label{emp-process} (\textit{Almost sure controlling empirical process by difference in learnt and true functions}) $$\Gamma_n^{EP} = O\left(\|\hat \psi_n -\psi\|_{L_2(\P)}\sqrt{\dfrac{\log\log n}{n}}\right).$$
\end{lemma}
\begin{proof}
Using law of iterated logarithm, following Step 1 of Lemma A.4 of \cite{waudby2021time}. 
\end{proof}

\textit{Remark: If $\hat\eta\in \T_n$, then $\|\hat\psi_n - \psi\|\le r_n'$ as in Assumption~\ref{regularity}, which based on Lemma \ref{emp-process} thus controls $\Gamma_n^{EP}$ at a pre-specified rate.}

\begin{lemma}
\label{bias-control}
(\textit{Controlling the bias term}) Under Assumption~\ref{dml-linear} and~\ref{regularity}, $$\Gamma^B_n = O(\lambda_n +\lambda_n')$$
\end{lemma}

\begin{proof} We proceed as in \cite{chernozhukov2018double}, using Taylor expansion and the Neyman orthogonality property of $\psi$.  Define $$g(r) = \P[\psi(W;\theta_0, \eta_0 + r(\hat\eta - \eta_0)] - \P\psi.$$ Then, it is easy to see that $\Gamma_n^B = g(1)$, and by Taylor expansion we obtain, 
$$g(1) = g(0) + g'(0) + g''(\tilde r)/2,\ \ \ \text{for some $\tilde r\in (0,1).$}$$
Clearly $g(0) = \P[\psi(W;\theta_0,\eta_0)] = 0$ and, on the event that $\hat\eta\in \T_n$ (which can not happen only finitely many times), 
$$\|g'(0)\| = \left\|\left.\dfrac{d}{dr}\P[\psi(W;\theta_0,\eta_0+r(\hat\eta-\eta_0))]\right|_{r=0}\right\|\le \sup_{\eta\in\T_n}\left\|\left.\dfrac{d}{dr}\P[\psi(W;\theta_0,\eta_0+r(\eta-\eta_0))]\right|_{r=0}\right\|\le \lambda_n.$$ 

Also, when $\hat\eta\in \T_n$, $$\|g''(\tilde r)\|\le \sup_{r\in (0,1)} \|g''(r)\|\le \lambda_n'.$$
Since $\hat\eta\in \T_n$ all but finitely many times, we have, $$\Gamma_n^B = g(1) = O(\lambda_n +\lambda_n').$$ 
\end{proof}

\samplesplitting*
\begin{proof}
By definition, we know that 
$$\hat\theta = -[\P_T \hat \psi_a]^{-1}\P_T \hat\psi_b,$$ and thus, 
\begin{align*}
    \hat\theta_0 - \theta_0 &= -[\P_T \hat \psi_a]^{-1}(\P_T\hat\psi_a \theta_0 + \P_T\hat\psi_b) 
    = -[J_0 + R_{n,1}]^{-1}\P_T\psi(W;\theta_0,\hat\eta),
\end{align*}
using the linearity of $\psi$ function in $\theta$, and letting $J_0 = \P\psi_a$ and $R_{n,1} = \P_T\hat \psi^a - \P\psi^a$. Now, 
$$(J_0+R_{n,1})^{-1} - J^{-1} = (J_0+R_{n,1})^{-1}(J_0 - (J_0 + R_{n,1}))J_0^{-1} = -(J_0+R_{n,1})^{-1}R_{n,1}J_0^{-1}$$
Next, using the same decomposition as Lemma~\ref{decomposition} on $\P_T\hat\psi^a - \P\psi^a$, and noting $(\P_T-\P)\psi = O\left(\sqrt{\dfrac{\log\log n}{n}}\right)$, we have, 
\begin{align*} R_{n,1} &= O\left(\sqrt{\dfrac{\log\log n}{n}}\right) + O\left(\|\hat \psi^a-\psi^a\|\sqrt{\dfrac{\log\log n}{n}}\right) + O(r_n) \\
&= O\left((1+\|\hat \psi^a-\psi^a\|)\sqrt{\dfrac{\log\log n}{n}}+r_n\right)\\
&= O\left(\sqrt{\dfrac{\log\log n}{n}}+r_n\right),
\end{align*}
noting that $\|\hat\psi^a-\psi^a\||\theta_0|\le \|\hat\psi-\psi\|=O(1)$.
Also, $J_0$ has all eigen-values bounded below by $c_0$, thus making $\|J_0^{-1}\|=O(1)$. Hence, 
$$\|(J_0+R_{n,1})^{-1} - J_0^{-1}\| \le \|(J_0+R_{n,1})^{-1}\|\|R_{n,1}\|\|J_0^{-1}\| = O\left(\sqrt{\dfrac{\log\log n}{n}}+r_n\right).$$
Next, by Lemma~\ref{decomposition}, Lemma~\ref{emp-process} and Lemma~\ref{bias-control}, 
$$\P_T\hat\psi = (\P_T-\P)\psi + O\left(r_n'\sqrt{\dfrac{\log\log n}{n}}\right) + O(\lambda_n + \lambda_n'),$$
and is thus $o(1)$ by assumption. Finally, 
$$\sigma^2 = J_0^{-1}\P[\psi\psi'](J_0^{-1})^T$$ has all eigen values bounded from below by $c_0/c_1^2$, making $\|\sigma^{-1}\|\le c_1/\sqrt{c_0} = O(1)$.
Combining everything, and noting that $\P\psi = 0$, we have, 
\begin{align*}
    \hat\theta_0-\theta_0 &= -[J_0+R_{n,1}]^{-1}\P_T\hat\psi \\
    &= -J_0^{-1}\P_T\hat\psi - [(J_0+R_{n,1})^{-1}-J_0^{-1}]\P_T\hat\psi \\
    &= -J_0^{-1}\P_T\psi -\left\{\Gamma_n^{EP} +\Gamma_n^B + [(J_0+R_{n,1})^{-1}-J_0^{-1}]\P_T\hat\psi\right\}\\
    &= -\dfrac 1T\sum_{i=1}^T J_0^{-1}\psi(W_i;\theta_0,\eta_0) + O\left((r_n+r_n')\sqrt{\dfrac{\log\log n}{n}} + \lambda_n +\lambda_n'\right) 
\end{align*} 
and left multiplying by $\sigma^{-1}$ completes the proof.
\end{proof}

\subsection{Proof of Theorem~\ref{var-est}}
\variance*
\begin{proof}
We proceed by decomposing $\P_T\hat\psi\psi^T$ into three component terms in a similar but not the exact same flavor as in Lemma \ref{decomposition}: 
\begin{align*}
    \P_T\hat\psi\hat\psi^T &= \P_T(\hat\psi - \psi + \psi)(\hat\psi - \psi + \psi)^T\nonumber \\ &= \P_T(\hat\psi-\psi)(\hat\psi-\psi)^T + \P_T\psi\psi^T + \P_T(\hat\psi-\psi)\psi^T + \P_T\psi(\hat\psi-\psi)^T. 
\end{align*}
The second term is expected to have a non-zero finite limit $\P\psi\psi^T$, while the other terms should be small. Subtracting of the potential limit $\P \psi\psi^T$ from the second term and applying triangle inequality, we obtain
\begin{align}
    \|\P_T\hat\psi\hat\psi^T - \P\psi\psi^T\|&\le \|\P_T(\hat\psi-\psi)(\hat\psi-\psi)^T\| + \|(\P_T - \P)\psi\psi^T\| + \|\P_T(\hat\psi-\psi)\psi^T\| + \|\P_T\psi(\hat\psi-\psi)^T\|\nonumber \\ 
    & \le \|\P_T(\hat\psi-\psi)(\hat\psi-\psi)^T\| + \|(\P_T - \P)\psi\psi^T\| + 2\|\P_T(\hat\psi-\psi)(\hat\psi-\psi)^T\|^\frac12\|\P_T\psi\psi^T\|^\frac12 ,\label{vardiff}
\end{align}
because $\|\P_T[uv^T]\|\le \|\P_T[uu^T]\|^\frac12\|\P_T[vv^T]\|^\frac12$. Now, one can bound the first term in \eqref{vardiff} as $\|\P_T(\hat\psi-\psi)(\hat\psi-\psi)^T\|\le \P_T(\|\hat\psi-\psi\|^2_2)$. Also, conditional on the training set, $(\hat\psi-\psi)$ are iid, and hence $\P_T\|\hat\psi-\psi\|_2^2$ is a conditional reverse submartingale. By \cite{manole2023martingale} Theorem 2, we have, 
$$\lim_{M\to\infty} \P\left(\exists\ T\ge 1: \P_T\|\hat\psi-\psi\|_2^2\ge M\P\|\hat\psi-\psi\|^2_2\right)\le \lim_{M\to\infty}\dfrac 1M = 0,
$$
leading to $\P_T\|\hat\psi-\psi\|_2^2 = O(\P\|\hat\psi-\psi\|^2_2) = O(r_n'^2)$. Next, the second term in Equation \eqref{vardiff} $(\P_T-\P)\psi\psi^T = O\left(\sqrt{\dfrac{\log\log n}{n}}\right)$ if $q\ge 4$ by Law of iterated logarithm, and $o(n^{-(1-\frac 2q)})$ if $2\le q<4$ by  Marcinkiewicz-Zygmund strong law of large numbers (\cite{Marcinkiewicz1937}). Finally, noting that the third term in Equation \eqref{vardiff} is a product of $O(r_n')$ and $O(1)$ terms, we can thus obtain,
\begin{equation} \P_T\hat\psi\hat\psi^T = \P\psi\psi^T + O(r_n'^2) + O\left(n^{-(1-\frac 2q)}\bigvee \sqrt{\dfrac{\log\log n}{n}}\right) + O(r_n')O(1) = O\left(n^{-(1-\frac 2q)}\bigvee \sqrt{\dfrac{\log\log n}{n}}+r_n'\right). \label{varorder}
\end{equation}
Finally, we have already shown in the proof of theorem~\ref{dml-linear}, that with $\hat J_0 = J_0 + R_{n,1}$, we have, 
\begin{equation}
\|\hat J_0^{-1} - J_0^{-1}\| = O\left(\sqrt{\dfrac{\log\log n}{n}} + r_n\right). \label{jhat}
\end{equation}
Thus, using \eqref{jhat} and \eqref{varorder},
\begin{align*} \hat\sigma^2 &= \hat J_0^{-1}\P_T\hat\psi\hat\psi^T\hat (\hat J_0^{-1})' = J_0^{-1}\P_T\hat\psi\hat\psi^T\hat ( J_0^{-1})' + O\left(\sqrt{\dfrac{\log\log n}{n}} + r_n\right) \\
&= J_0^{-1}\P\psi\psi^T ( J_0^{-1})' + J_0^{-1}(\P_T\hat\psi\hat\psi^T-\P\psi\psi^T)\hat ( J_0^{-1})'+ O\left(\sqrt{\dfrac{\log\log n}{n}} + r_n\right)\\
&= \sigma^2 + O\left(n^{-(1-\frac 2q)}\bigvee\sqrt{\dfrac{\log\log n}{n}} + r_n + r_n'\right).
\end{align*} 
\end{proof}

\subsection{Proof of Theorem~\ref{csdml}}

\begin{lemma} 
    \label{csgaussian}
    (Confidence set for multivariate Brownian motion) If $\{B_t\}_{t\ge 1}$ is a $p$-dimensional Brownian motion, then with probability at least $1-\alpha$, 
    $$\forall t\ge 1,\ \dfrac{\|B_t\|^2}{t^2}<\dfrac{2(t\rho^2+1)}{t^2\rho^2}\log\dfrac{(t\rho^2+1)^{\frac p2}}{\alpha}.$$
\end{lemma}
\begin{proof}
 Let $\{B_t\}_{t\ge 0}$ be a $p$-dimensional Brownian motion, and consider 
$$M_t(\lambda) = \exp\left(\lambda^TB_t - \dfrac{\|\lambda\|^2 t}{2}\right).$$ 
It is easy to see that it is a non-negative martingale with respect to the natural filtration with $M_0(\lambda) = 1$, and thus by Fubini's theorem, mixing $M_t(\lambda)$ with any distribution of $\lambda$ on $\R^p$, would also result in a non-negative martingale with respect to the same filtration and initial value. In particular, we choose $f(\lambda;0,\rho^2I_p)$ to be the centered Gaussian distribution with variance $\rho^2I_p$, resulting in 
\begin{align*}
    M_t &= \int_{\R^p} M_t(\lambda)f(\lambda;0,\rho^2I_p)\,d\lambda \\
    &= \dfrac{1}{(2\pi\rho^2)^{\frac p2}}\int \exp\left(\lambda^T B_t -\dfrac{\|\lambda\|^2 t}{2}\right) \exp\left(-\dfrac{\|\lambda\|^2}{2\rho^2}\right) \,d\lambda= \dfrac{1}{(t\rho^2+1)^{\frac p2}}\exp\left(\dfrac{\rho^2\|B_t\|^2}{2(t\rho^2+1)}\right).
\end{align*}
Now by Ville's inequality and expanding out the expression for $M_t$,
\begin{align*}
    1-\alpha&\le \P(\forall t\ge 1, M_t<1/\alpha) = \P\left(\forall t\ge 0, \dfrac{\|B_t\|^2}{t^2}< \dfrac{2(t\rho^2+1)}{t^2\rho^2}\log\dfrac{(t\rho^2+1)^{\frac p2}}{\alpha}\right).
\end{align*} \end{proof}

\confidence*
\begin{proof}
(b) follows trivially from Corollary~\ref{seqdml-all}, and \cite{waudby2021time} Theorem 2.2.

To prove (a), define $Y_i = \sigma^{-1}J_0^{-1}\psi(W_i;\theta_0,\eta_0)$, then by Assumption~\ref{neyman-orthogonal}(a), has expectation 0, and variance $I_d$ by definition of $\sigma^2$. Thus, by strong approximation theorem in $\R^d$ (\cite{strassen1964invariance, einmahl1989extensions, strassen1967almost}), there exist iid Gaussian random variables $\{G_i\}_{i=1}^n$ such that
$$\dfrac 1n\sum_{i=1}^n Y_i = \dfrac{1}{n}\sum_{i=1}^n G_i + \varepsilon_n,$$ where $\varepsilon_n = o(\sqrt{\log\log n/n})$. The right hand side can also be replaced by a Weiner process. 
Thus, by Lemma~\ref{csgaussian} 
$$\forall n\ge 1,\ \left\|\dfrac 1n \sum_{i=1}^n Y_i\right\|\le \underbrace{\sqrt{\dfrac{2(n\rho^2+1)}{n^2\rho^2}\log\dfrac{(n\rho^2+1)^\frac d2}{\alpha}}}_{\bar{\mathcal{B}}_n} + \|\varepsilon_n\| \ \ \text{ with probability }\ge 1-\alpha. $$
The proof is finally completed noting that $\bar{\mathcal{B}}_n$ is $O\left(\sqrt{\dfrac{\log n}n}\right)$ and the error in approximation terms in Corollary~\ref{seqdml-all} as well as $\|\varepsilon_n\|$ are all $o\left(\sqrt{\dfrac{\log n}n}\right)$, suppressing the dependence on $d$ for the approximation errors. \end{proof}

\section{Proof of Theorem~\ref{partial-theorem}}
\label{partialproof}

\partialthm*
\begin{proof}
    
Let $\mathcal{T} := \{\eta =(g_1,e,\nu): \eta\text{ is measurable and }1\le \nu(x)\le \Gamma, \varepsilon\le e(x)\le 1-\varepsilon\}$, such that $\eta_1\in\mathcal{T}$. Following \cite{yadlowsky2018bounds}, the proof continues in two parts, via the verification of Assumptions~\ref{neyman-orthogonal} and~\ref{regularity}. We closely follow the proof of \cite{yadlowsky2018bounds} throughout.

\textit{Verification of Assumption~\ref{neyman-orthogonal}}
(a) $\E[\psi(W; \mu_1^-,\eta_1) = 0$ holds by construction and the identification condition. (b) is evident, with $\psi^a \equiv -1$. (c) \cite{yadlowsky2018bounds} proves the twice Gauteax differentiability on $\mathcal{T}$ in Section E.2.1. (d) Orthogonality follows from Section E.2.1 and Lemma 3 of \cite{yadlowsky2018bounds}. (e) The identification condition holds as $J_0 = -1$.

\textit{Verification of Assumption~\ref{regularity}}
(a) By assumption~\ref{partial-ass}(d) we can obtain a sequence $a_n\to 0$ such that for each $n$, 
$$\|\hat g(\cdot) - g_1(\cdot)\|_2 \le a_n\left(\dfrac{\log\log n}{n}\right)^\frac 14, \|\hat \hat e(\cdot) - e_1(\cdot)\|_2 \le a_n\left(\dfrac{\log\log n}{n}\right)^\frac 14, \|\hat \nu(\cdot) - \nu_1(\cdot)\|_2 \le a_n\left(\dfrac{\log\log n}{n}\right)^\frac 14$$ for all but finitely many $n$. In fact, $a_n$'s can be chosen so that the above holds when estimated from a subset of $n-T$ samples (instead of from $n$ samples) where $T/n\to \frac 1K$ a.s. Also, assumption~\ref{partial-ass}(c) implies $\exists C_1$ such that $\|\hat\eta_1-\eta_1\|_{\infty}\le C_1$ for all but finitely many $n$. 

Let $$\mathcal{T}_n := \left\{\eta = (g_1,e,\nu): \|\eta -\eta_1\|_2\le a_n\left(\dfrac{\log\log n}{n}\right)^\frac 14, \|\eta -\eta_1\|_\infty \le C_1, \varepsilon\le e\le 1-\varepsilon, 1\le \nu\le \Gamma\right\}.$$
By the strong consistency of $\hat\eta_1$ in Assumption~\ref{partial-ass}(c), we thus obtain $\hat\eta\in \mathcal{T}_n$ for all but finitely many $n$, thus verifying Assumption~\ref{regularity}(a). 

For (b), $m_n' = 1$ trivially, and the bounds on $m_n$ have been proved in \cite{yadlowsky2018bounds} Section E.2.2.

To verify Assumption~\ref{regularity}(c), first note that $\psi^a \equiv -1\implies r_n = 0$. Also, \cite{yadlowsky2018bounds} Section E.2.2 shows that 
\begin{align*}
    \sup_{\eta\in\mathcal{T}_n}\E[(\psi(W;\mu_1^-, \eta_1) - \psi(W;\mu^-, \eta))^2]^\frac 12 &\le \sup_{\eta\in \mathcal{T}_n} \|g(\cdot) - g_1(\cdot)|_2 + \dfrac{\Gamma(1-\varepsilon)}{\varepsilon}\sup_{\eta\in \mathcal{T}_n}\|\nu - \nu_1\|_2 
     + \dfrac{2\Gamma^2(1-\varepsilon)}{\varepsilon}\times\\ &\hspace{2em}(\|g(\cdot)\|_2 + \E[Y(1)^2])\left(\|e(\cdot) - e_1(\cdot)\|_2 + \|\nu(\cdot) - \nu_1(\cdot)\|_2\right),
\end{align*}
and thus $r_n'\le \tilde C_2 a_n$ for large enough $\tilde C_2$. Finally, to prove the double derivative constraint, following the proof of Lemma E.1 in \cite{yadlowsky2018bounds}, Section E.2.3, for $r\in (0,1)$ one can obtain,
$$\left|\dfrac{d}{dr^2}\E(W;\mu_1^-, \eta_1 + r(\eta-\eta_1))\right|\le C(\|e(\cdot) - e_1(\cdot)\|^2_2 + \|g(\cdot) - g_1(\cdot)|^2_2 + \|\nu(\cdot) - \nu_1(\cdot)\|^2_2),$$ and thus by construction of $\mathcal{T}_n$, $\|\eta-\eta_1\|\le a_n\left(\dfrac{\log\log n}{n}\right)^\frac 14$ we obtain $\lambda_n'\le a_n\sqrt{\dfrac{\log\log n}{n}}$. 
Finally for Assumption~\ref{regularity}(d), it can be seen easily that 
$\E[\psi(W;\mu_1^-, \eta_1)^2|A=1,X=x] \ge \V(Y(1)|A=1, X=x)$ and taking expectations over $X$ gives us the desired condition.

Thus by Theorem~\ref{csdml}, we have our required result.
\end{proof}

\section{Proof of Theorem~\ref{iv-theorem}}
\label{lateproof}

\iv*
\begin{proof}
    
Let $\mathcal{T}:= \{\eta= (g_t,g_c,m_t,m_c,e): \eta \text{ is measurable and }\varepsilon \le e(X), m_a(X)\le 1-\varepsilon, \|g_a\|_q = O(1)\ \forall a=0,1\}$ such that $\eta_0\in\mathcal{T}$. We closely follow the steps in \cite{chernozhukov2018double} and \cite{ma2023identification}. The proof follows by verification of Assumption~\ref{neyman-orthogonal} and~\ref{regularity}, and thus applying Theorem~\ref{csdml}.

\textit{Verification of Assumption~\ref{neyman-orthogonal}}: (a) $\E[\psi(W; \theta,\eta)] = 0$ follows by construction and the identification condition. (b) is evident from the form in \eqref{iv-functional}. (c) The regularity and boundedness conditions in Assumption~\ref{iv-ass} ensure twice Gateuax differentiability. (d) Neyman-orthogonality has been shown in \cite{ma2023identification} Appendix B.2 Step 1. (e) The identification criteria follows by the identification conditions in Assumption~\ref{iv-ass}.

\textit{Verification of Assumption~\ref{regularity}}:
By Assumption~\ref{iv-ass}(d), we can find a sequence $a_n$ such that $\|\hat\eta - \eta_0\|_2\le a_n$ and $\|\hat e- e\|_2\left(\sum_{a\in \{t,c\}} \|\hat g_a(\cdot) - g_{a,0}(\cdot)\|_2 + \sum_{a\in \{t,c\}} \|\hat m_a(\cdot) - m_{a,0}(\cdot)\|_2\right) \le a_n\sqrt{\dfrac{\log\log n}{n}} $ for all but finitely many $n$. In fact, $a_n$'s can be chosen so that the above holds when estimated from a subset of $n-T$ samples (as opposed to $n$) where $T/n\to 1/K$ a.s. Also, by the boundedness conditions on Theorem~\ref{iv-ass}, $\|\hat\eta_1-\eta\|_\infty\le C$ for some $C<\infty$.
Let 
$$\begin{aligned} \T_n&:= \left\{\vphantom{\sum_{a=t,c}}\eta = (g_t,g_c,m_t,m_c, e): \|\eta-\eta_0\|_2\le a_n, \|\eta -\eta_0\|_\infty \le C_1, \varepsilon\le e\le 1-\varepsilon, \|\eta-\eta_{0}\|_q\le C, \right.\\ &\left.\hspace{3em}\|e-e_0\|_2\left(\sum_{a=t,c} \|g_a-g_{a,0}\|+\|m_a-m_{a,0}\|\right)\le a_n\sqrt{\dfrac{\log\log n}{n}}\right\}. 
\end{aligned}$$ Then $\eta_0\in\mathcal{T}_n$ and $\hat\eta \in\mathcal{T}_n$ for all but finitely many $n$. This verifies Assumption~\ref{regularity}(a).

Assumption~\ref{regularity}(b) has been verified in Step 4 of Appendix B.2 of \cite{ma2023identification}. 
To verify Assumption~\ref{regularity}(c), $r_n\le a_n$ follows by noting that $\sup_{\eta\in\mathcal{T}_n} \|m_t-m_{t,0}\|_1, \sup_{\eta\in\mathcal{T}_n} \|m_c-m_{c,0}\|_1$ and $\sup_{\eta\in\mathcal{T}_n} \|e-e_0\|_1$ are all $O(a_n)$. Also, \cite{ma2023identification} in Step 4 of Appendix B.2 show that 
\begin{align*}
    \sup_{\eta\in\mathcal{T}_n}(\E[\|\psi(W;\theta_0,\eta) - \psi(W;\theta_0,\eta_0)\|^2])^\frac 12 &\le \sum_{a\in \{t,c\}} \|g_a(X)-g_{a,0}(X)\|_2 + 
     |\theta_0|\sum_{a\in \{t,c\}} \|m_a(X) - m_{a,0}(X)\|_2 \\
     &+ \dfrac 1{\varepsilon^2}\left(\sum_{a\in {t,c}} \|g_a(X)-g_{a,0}(X)\|_2 + C_1\|e(X) - e_0(X)\|_2\right) \\
     &+ \dfrac {|\theta_0|}{\varepsilon^2}\left(\sum_{a\in {t,c}} \|m_a(X)-m_{a,0}(X)\|_2 + C_2\|e(X) - e_0(X)\|_2\right),
\end{align*}
for some constants $C_1$ and $C_2$. Thus by the strong consistency assumptions, $r_n\le a_n$. For the final step of verification of Assumption~\ref{regularity}(c), it can be seen by twice Gateuax differentiation of the functional and taking the respective expectation that 
\begin{align*}
    \sup_{r\in (0,1);\eta\in\mathcal{T}_n}\left\|\dfrac{d}{dr^2}\E(W;\theta_0, \eta_0 + r(\eta-\eta_0))\right\|&\le C'\left(\dfrac{1}{1-\varepsilon}(\|g_t - g_{t,0}\|_2+|\theta_0|\|m_t-m_{t,0}\|_2)\|e-e_0\|_2\right. \\ &\hspace{3em}  + \left.\dfrac{1}{\varepsilon}(\|g_c-g_{c,0}\|_2+|\theta_0|\|m_c-m_{c,0}\|_2)\|e-e_0\|_2\right), 
\end{align*} and hence $\lambda_n' \le a_n\sqrt{\dfrac{\log\log n}{n}}$.
Finally Assumption~\ref{regularity}(d) follows from Step 2 in Appendix B.2 of \cite{ma2023identification}. 
\end{proof}

\end{document}